\documentclass[11pt,reqno]{amsart}
\usepackage{tikz}
\usepackage{subfig}
\usepackage{epstopdf}
\usepackage{epsfig}
\usepackage{math}

\usepackage{tikz}
\usetikzlibrary{shapes,arrows}

\tikzstyle{decision} = [diamond, draw, fill=blue!20, 
    text width=4.5em, text badly centered, node distance=3cm, inner sep=0pt]
\tikzstyle{block} = [rectangle, draw, fill=blue!20, 
    text width=5em, text centered, rounded corners, minimum height=4em]
\tikzstyle{line} = [draw, -latex']
\tikzstyle{cloud} = [draw, ellipse,fill=red!20, node distance=3cm,
    minimum height=2em]

\usetikzlibrary{positioning}
\tikzset{main node/.style={circle,fill=blue!20,draw,minimum size=1cm,inner sep=0pt},  }

\begin{document}
\title{Diffusion Hypercontractivity via Generalized Density Manifold}
\author{Wuchen Li}
\email{wcli@math.ucla.edu}
\begin{abstract}
We prove a one-parameter family of diffusion hypercontractivity and present the associated Log-Sobolev, Poincar{\'e} and Talagrand inequalities. A mean-field type Bakry-{\'E}mery iterative calculus and volume measure based integration formula (Yano's formula) are presented. Our results are based on the interpolation among divergence functional, generalized diffusion process, and generalized optimal transport metric. As a result, an inequality among Pearson divergence (P), negative Sobolev metric ( $H^{-1}$) and generalized Fisher information functional (I), named P$H^{-1}$I inequality, is derived.  
\end{abstract}
\keywords{Information theory; Mean-field Bakry-{\'E}mery calculus; Generalized Log-Sobolev inequality; Generalized Poincar{\'e} inequality; Generalized Talagrand inequality; Generalized Yano's formula.}
\thanks{Wuchen Li is supported by AFOSR MURI (Grant No.~18RT0073).}
\maketitle
\section{Introduction}
Diffusion hypercontractivity plays essential roles in functional inequalities \cite{BGL,LSI} and information theory \cite{CoverThomas1991_elements,CsiszarShields2004_information}. It is often used to estimate the convergence rates of Markov chain Monte Carlo methods. Among these studies, Bakry--{\'E}mery criterions \cite{BE} provide sufficient conditions for showing convergences rates of diffusion processes and related inequalities. 
Recently, optimal transport provides the other viewpoint on this topic \cite{Villani2009_optimal}. In this viewpoint, the probability density space is embedded with an infinite-dimensional Riemannian metric, named Wasserstein metric \cite{ otto2001}. Here the density space with Wasserstein metric is named density manifold \cite{Lafferty}. Following this metric viewpoint, the diffusion hypercontractivity, in particular, the 
Bakery--{\'E}mery criterions, can be derived by Hessian operators of divergence functionals in density manifold \cite{OV}. This study has been extended to general base metric spaces \cite{Lott_Villani, strum}. On the other hand, the relation between local behavior of diffusion hypercontractivity (such as Poincar{\'e} inequality) and integral formula, known as Yano's formula \cite{yano1}, has been discovered in \cite{ChowLiZhou2018_entropy,Li-thesis}. It reveals the connection between integration formula on the base manifold and Riemannian calculus in density manifold.  

In this paper, we study the hypercontractivity of generalized diffusion processes, named Hessian transport stochastic differential equations \cite{Li}. See related information theory background in \cite{ZozorBrossier2015_debruijn}. Following the generalized (mobility) density manifold proposed in \cite{C1,O1}, the density manifold's Riemannian calculus \cite{LiG1} and geometric insights of inequalities provided in \cite{OV}, we introduce the generalized Bakry--{\'E}mery criterions in Theorem \ref{thm}. These criterions provide sufficient conditions for showing convergence rates of generalized diffusion processes and establishing generalized Log-Sobolev and Talagrand inequalities. In addition, a generalized Yano's formula in Theorem \ref{cor2} is derived, which provides a reference measure related integral formulas.  Using Yano's formula, we establish the generalized Poincar{\'e} inequality in Corollary \ref{GPI}. More importantly, a P$H^{-1}$I inequality is presented in Theorem \ref{cor3}. 

In literature, the generalized optimal transport metric has been proposed in \cite{O1} and many groups have studied associated generalized functional inequalities \cite{Gentil,O2}. Firstly, \cite{O2} studies functional inequalities for the classical Kolmogorov-Fokker-Planck equation, where the Bakry--{\'E}mery criterions are classical and generalized optimal transport metrics depend on the reference measure. Here we study transport metrics related stochastic processes and then build new functional inequalities. We introduce new Bakry--{\'E}mery criterions, in which the metric of density manifold in inequalities do not depend on the reference measure. For example, we obtain several functional inequalities related to the classical $H^{-1}$ metric. In addition, \cite{Gentil} formulates divergence functional related inequalities for a type of drift-diffusion processes. In this study, they apply the classical Bakry-{\'E}mery iterative calculus. While in this paper, we introduce a new mean field Bakry-{\'E}mery iterative calculus.

This paper is organized as follows. In section \ref{section2}, we state the main result of this paper. We establish hypercontractivity for generalized drift-diffusion processes and prove several functional inequalities. A generalized Yano's formula is also derived. In section \ref{section3}, we formulate the primary tool of the proof. In section \ref{section4}, the proof is presented. In section \ref{section5}, a generalized Bakry-{\'E}mery iterative calculus is presented.  
\section{Main result}\label{section2}
Given a compact and smooth Riemannian manifold $(M, (\cdot,\cdot))$ without boundary. Denote its volume form by $dx$, the Ricci curvature tensor by $\textrm{Ric}$, the gradient, divergence, Laplacian operators by $\nabla$, $\nabla\cdot$, $\Delta$ respectively,  and the Hessian operator by $\textrm{Hess}$.  

Given a reference probability density function $\mu\in C^{\infty}(M)$ with $\inf_{x\in M}\mu(x)>0$, consider the {\em $\gamma$--drift diffusion process}  
\begin{equation*}
dX_{\gamma, t}=-\frac{\gamma}{\gamma-1}\nabla\mu(X_{\gamma, t})^{\gamma-1}dt+\sqrt{2\mu(X_{\gamma, t})^{\gamma-1}}dB_t,
\end{equation*}
where $B_t$ is the standard Brownian motion in $M$ with the infinitesimal generator  
\begin{equation*}
L_\gamma \Phi=(\frac{\gamma}{\gamma-1}\nabla\mu^{\gamma-1}, \nabla \Phi)+\mu^{\gamma-1}\Delta \Phi, \quad \Phi\in C^{\infty}(M).
\end{equation*}
Consider the {\em $\gamma$--divergence functional}\footnote{It is often named $\alpha$--divergence with $\gamma=\frac{3-\alpha}{2}$. We use notation $\gamma$ for the simplicity of presentation.}
\begin{equation*}
\mathcal{D}_\gamma(\rho\|\mu)=\int f(\frac{\rho}{\mu})\mu dx,
\end{equation*}
where $f\colon [0,\infty)\rightarrow \mathbb{R}$ has the form 
\begin{equation*}
f(\rho)=\begin{cases}
\frac{1}{(1-\gamma)(2-\gamma)}(\rho^{2-\gamma}-1) & \gamma\neq 1, \gamma\neq 2\\
\rho\log\rho &\gamma=1\\
-\log\rho & \gamma=2.
\end{cases}
\end{equation*}
Consider the {\em $\gamma$--Fisher information functional}
\begin{equation*}
\mathcal{I}_\gamma(\rho\|\mu)=\int \|\nabla \log \frac{\rho}{\mu}\|^2 \rho^{\gamma}\mu^{2\gamma-2} dx.
\end{equation*}
Consider the {\em $\gamma$--Wasserstein distance}\footnote{This notation of Wasserstein-distance is first studied in \cite{O1}. When $\gamma=1$, we notice that the notation of $\mathcal{W}_{1}$ represents the classical $L^2$-Wasserstein distance, not the $L^1$-Wasserstein distance. }
\begin{equation*}
\mathcal{W}_\gamma(\rho,\mu)=\inf_{\Phi} \Big\{\int_0^1\sqrt{\int \|\nabla\Phi_t\|^2\rho_t^\gamma dx} dt\colon \partial_t\rho_t+\nabla\cdot(\rho_t^\gamma\nabla\Phi_t)=0,~\rho_0=\rho,~\rho_1=\mu\Big\},
\end{equation*}
where $\rho_t=\rho(t,x)$, $\Phi_t=\Phi(t,x)$ and the infimum is over all potential function $\Phi\in [0,1]\times M\rightarrow \mathbb{R}$. 

We next provide sufficient conditions for the hypercontractivity of $\gamma$--drift diffusion process and functional inequalities among $\gamma$--divergence, $\gamma$--Fisher information and $\gamma$--Wasserstein distance. 
\begin{theorem}[Generalized hypercontractivity]\label{thm}
Let $\gamma\in [0,1]$, if there exists a constant $\kappa>0$, such that
\begin{equation}\label{condition}
\mu^{\gamma-1}\textrm{Ric}-\frac{1}{\gamma-1}\textrm{Hess}\mu^{\gamma-1}-\Delta \mu^{\gamma-1}+\frac{1}{8}\gamma(\gamma-1) \|\nabla\log\mu\|^2\mu^{\gamma-1} \succeq \kappa.
\end{equation}
 Let $\rho_0$ be a smooth initial distribution and $\rho_t$ be the probability density function of $\gamma$--drift diffusion process, then 
\begin{equation}\label{hyper}
\mathcal{D}_\gamma (\rho_t\|\mu) \leq e^{-2\kappa t} \mathcal{D}_\gamma(\rho_0\|\mu). 
\end{equation}
Moreover, for any smooth probability density function $\rho\in C^{\infty}(M)$ with $\inf _{x\in M}\rho(x)>0$, the generalized Log-Sobolev inequality holds
\begin{equation}\label{inequality}
\mathcal{D}_\gamma(\rho\|\mu) \leq \frac{1}{2\kappa}\mathcal{I}_\gamma(\rho\|\mu),
\end{equation}
and the generalized Talagrand inequality holds 
\begin{equation}\label{Talagrand}
\mathcal{W}_\gamma(\rho, \mu)\leq\sqrt{\frac{2\mathcal{D}_\gamma(\rho\|\mu)}{\kappa}}.
\end{equation}
\end{theorem}

\begin{example}[Kullback--Leibler divergence]
Consider $\gamma=1$, $f(\rho)=\rho\log \rho$. Here the $\mathcal{D}_1=\int \rho \log\frac{\rho}{\mu}dx$ forms the classical Kullback--Leibler divergence function (relative entropy),
and $\mathcal{I}_1=\int \|\nabla\log\frac{\rho}{\mu}\|^2\rho dx$ is the classical relative Fisher information, $dX_{1,t}=-\nabla\log\mu(X_{1,t})dt+\sqrt{2}B_t$ is the classical Langevin process, and $\mathcal{W}_1$ is the classical $L^2$-Wasserstein distance.   
Here the condition \eqref{condition} forms 
\begin{equation*}
\textrm{Ric}-\textrm{Hess}\log\mu \succeq \kappa, \quad \kappa>0,
\end{equation*}
which is the classical Bakry--{\'E}mery criterion. Under this condition, the distribution of drift diffusion process $X_{1,t}$ converges to $\mu$; the Log--Sobolev inequality \eqref{inequality} holds 
\begin{equation*}
\int \rho\log\frac{\rho}{\mu} dx\leq \frac{1}{2\kappa}\int \|\nabla\log\frac{\rho}{\mu}\|^2\rho dx;
\end{equation*}
and the Talagrand inequality holds 
\begin{equation*}
\mathcal{W}_1(\rho, \mu)\leq \sqrt{\frac{2\mathcal{D}_1(\rho\|\mu)}{\kappa}}.
\end{equation*}
Here, if $M$ is Ricci flat, i.e. $\textrm{Ric}=0$ and we denote $\mu(x)=e^{-V(x)}$, the condition \eqref{condition} forms $\textrm{Hess}V\succeq\kappa$.
\end{example}
\begin{example}[Pearson divergence]
Consider $\gamma=0$, $f(\rho)=\frac{1}{2}(\rho^2-1)$. Here $\mathcal{D}_0=\frac{1}{2}\int (\frac{\rho}{\mu}-1)^2\mu dx$ is named Pearson divergence function, $\mathcal{I}_0=\int\|\nabla\log\frac{\rho}{\mu}\|^2\mu^{-2}dx$ is the $0$--Fisher information and 
$dX_{0,t}=\sqrt{2\mu^{-1}(X_t)}dB_t$ is the $0$--drift diffusion process. The condition \eqref{condition} forms 
\begin{equation*}
\mu^{-1}\textrm{Ric}+\textrm{Hess}\mu^{-1}-\Delta \mu^{-1} \succeq\kappa,\quad \kappa>0.
\end{equation*}
Under this condition, the distribution of drift diffusion process $X_{0,t}$ converges to $\mu$ and the generalized Log--Sobolev inequality \eqref{inequality} holds
\begin{equation*}
\frac{1}{2}\int (\frac{\rho}{\mu}-1)^2\mu dx \leq \frac{1}{2\kappa}\int\|\nabla\log\frac{\rho}{\mu}\|^2\mu^{-2}dx.
\end{equation*}
\end{example}

We also show a new integration identity, which is found in the proof of Theorem \ref{thm}. 
\begin{theorem}[Generalized Yano's formula]\label{cor2}
Denote $\Phi\in C^{\infty}(M)$, then
\begin{equation*}
\begin{split}
&\int \mu^{-1}\Big(\nabla\cdot(\mu^{\gamma}\nabla\Phi)\Big)^2 dx\\
=&\int \mu^{\gamma}\Big\{(\mu^{\gamma-1}\textrm{Ric}-\Delta \mu^{\gamma-1}-\frac{1}{\gamma-1}\textrm{Hess}\mu^{\gamma-1})(\nabla\Phi, \nabla\Phi)+\mu^{\gamma-1}\|\textrm{Hess}\Phi\|^2\\
&\qquad+\gamma(\gamma-1)\mu^{\gamma-1}\Big((\nabla\log\mu,\nabla\Phi)^2-\|\nabla\log\mu\|^2\|\nabla\Phi\|^2\Big)\Big\} dx.
\end{split}
\end{equation*}
\end{theorem}
\begin{remark}
When $\mu(x)$ is a uniform measure, i.e. $\mu(x)=1$, the above formula is the classical Yano's formula \cite{yano1}:
\begin{equation*}
\int (\Delta\Phi)^2 dx= \int \Big\{\textrm{Ric}(\nabla\Phi, \nabla\Phi)+\|\textrm{Hess}\Phi\|^2\Big\} dx.
\end{equation*}
When $\gamma=1$, it is the generalized Yano's formula studied in \cite{ChowLiZhou2018_entropy,Li-thesis}
\begin{equation*}
\int \mu^{-1}\Big(\nabla\cdot(\mu \nabla\Phi)\Big)^2 dx =\int \mu \Big\{ \textrm{Ric}(\nabla\Phi, \nabla\Phi)+\|\textrm{Hess}\Phi\|^2\Big\}dx. 
\end{equation*}
Our derivation extends these classical Yano's formulas with general volume measure $\mu$ and its power $\gamma$. For example, when $\gamma=0$, we obtain 
\begin{equation*}
\begin{split}
\int \mu^{-1}(\Delta\Phi)^2 dx=&\int\Big\{(\mu^{-1}\textrm{Ric}-\Delta \mu^{-1}+\textrm{Hess}\mu^{-1})(\nabla\Phi, \nabla\Phi)+\mu^{-1}\|\textrm{Hess}\Phi\|^2\Big\} dx.
\end{split}
\end{equation*}
\end{remark}
Later on, using the generalized Yano's formula, we prove the following Poincar{\'e} inequality. 
\begin{corollary}[Generalized Poincar{\'e} inequality]\label{GPI}
If there exists a constant $\lambda>0$, such that when $\gamma\in [0,1]$, 
\begin{equation*}
\mu^{\gamma-1}\textrm{Ric}-\Delta \mu^{\gamma-1}-\frac{1}{\gamma-1}\textrm{Hess}\mu^{\gamma-1}\succeq \lambda, 
\end{equation*}
or when $\gamma\in[1,\infty)\cup(-\infty, 0]$, 
\begin{equation*}
\mu^{\gamma-1}\textrm{Ric}-\Delta \mu^{\gamma-1}-\frac{1}{\gamma-1}\textrm{Hess}\mu^{\gamma-1}-\gamma(\gamma-1)\|\nabla\log\mu\|^2\mu^{\gamma-1}\succeq \lambda, 
\end{equation*}
then 
\begin{equation}\label{PI}
\int f^2\mu dx\leq \frac{1}{\lambda} \int \|\nabla f\|^2\mu^{\gamma}dx,
\end{equation}
for any $f\in C^{\infty}(M)$ with $\int f \mu dx=0$.
\end{corollary}
\begin{remark}
When $\gamma=1$, Corollary \ref{GPI} recovers the classical Poincar{\'e} inequality
\begin{equation*}
\int f^2\mu dx\leq \frac{1}{\lambda} \int \|\nabla f\|^2\mu dx.
\end{equation*}
\end{remark}
\begin{remark}
Here we derive the formulation of generalized Poincar{\'e} inequality by the approximation of generalized Log-Sobolev inequality. In details, denote 
$\rho=\mu+\epsilon h$, where $h=f\mu$ and $\int h dx=0$. The L.H.S. of \eqref{PI} comes from the Hessian metric tensor for the $\gamma$--divergence $\mathcal{D}_\gamma(\rho\|\mu)$:
\begin{equation*}
\mathcal{D}_\gamma(\mu+\epsilon h\|\mu)=\frac{\epsilon^2}{2} \int f^2\mu dx+o(\epsilon^2).
\end{equation*}
While the R.H.S. of \eqref{PI} is from the second order approximation in term of $\epsilon$ for the $\gamma$--relative Fisher information:
\begin{equation*}
\mathcal{I}_\gamma(\mu+\epsilon h \|\mu)=\epsilon^2 \int (\nabla f)^2\mu^\gamma dx+o(\epsilon^2).
\end{equation*}
\end{remark}

\begin{example}[Reverse Kullback--Leibler divergence]
Consider $\gamma=2$, $f(\rho)=-\log \rho$. Here $\mathcal{D}_2(\rho\|\mu)=-\int \mu\log\frac{\rho}{\mu} dx$ is named reverse Kullback--Leibler divergence function or Cross entropy. 
In this case, the condition in Corollary \ref{GPI} forms  
\begin{equation*}
\mu \textrm{Ric}-\Delta \mu-\textrm{Hess}\mu-2\|\nabla\log\mu\|^2\mu\succeq \lambda, \quad \lambda>0,
\end{equation*}
Under this condition, the generalized Poincar{\'e} inequality holds
\begin{equation*}
\int f^2\mu dx\leq \frac{1}{\lambda} \int \|\nabla f\|^2\mu^{2}dx,
\end{equation*}
where $f\in C^{\infty}(M)$ and $\int f\mu dx=0$. 
Again, if $M$ is Ricci flat, i.e. $\textrm{Ric}=0$, then the condition in Corollary \ref{GPI} forms $-\Delta \mu-\textrm{Hess}\mu-2\|\nabla\log\mu\|^2\mu\succeq \lambda$.
\end{example}

Last, notice the fact that when $\gamma=0$, the $\gamma$-Wasserstein distance is exactly the $H^{-1}$ distance: $$\mathcal{W}_0(\rho, \mu)=H^{-1}(\rho, \mu),$$ 
where $H^{-1}$ is the negative Sobolev distance between $\rho$ and $\mu$, i.e. 
$$H^{-1}(\rho, \mu)=\sqrt{\int \big(\rho-\mu, \Delta^{-1}(\rho-\mu)\big) dx}.$$ 
We next show an inequality among Pearson divergence (P), $H^{-1}$ metric and $0$--Fisher information (I), named $PH^{-1}I$ inequality. 
\begin{theorem}[Inequalities for $H^{-1}$ metric]\label{cor3}
Suppose $\mu^{-1}\textrm{Ric}+\textrm{Hess}\mu^{-1}-\Delta \mu^{-1} \succeq \kappa$, where $\kappa\in\mathbb{R}$, then the $PH^{-1}I$ inequality holds
\begin{equation*}
D_0(\rho\|\mu)\leq \sqrt{\mathcal{I}_0(\rho\|\mu)}H^{-1}(\rho,\mu)-\frac{\kappa}{2}H^{-1}(\rho,\mu)^2.
\end{equation*}
In addition, if $\kappa\geq 0$, then the $H^{-1}$-Talagrand inequality holds 
\begin{equation*}
H^{-1}(\rho, \mu)\leq\sqrt{\frac{2\mathcal{D}_0(\mu\|\nu)}{\kappa}},
\end{equation*}
\end{theorem}
\begin{remark}
The P$H^{-1}$I inequality is an analog of inequalities among $\mathcal{D}_1$ (H), Wasserstein-2 metric and $1$--Fisher information, known as HWI inequality; see details in \cite{OV}. %Shortly in Remark \ref{rmk9}, we demonstrate that the type of divergence, metric, information inequality only exists for $\gamma=0,1$.
\end{remark}
\begin{remark}
If $\kappa> 0$, the P$H^{-1}$I inequality shows 
\begin{equation*}
\mathcal{D}_0(\rho\|\mu)\leq \sqrt{\mathcal{I}_0(\rho\|\mu)}H^{-1}(\rho,\mu).
\end{equation*}
In addition, using the fact that $H^{-1}(\rho, \mu)\leq\sqrt{\frac{2\mathcal{D}_0(\mu\|\nu)}{\kappa}}$ and $\mathcal{D}_0(\rho\|\mu) \leq \frac{1}{2\kappa}\mathcal{I}_0(\rho\|\mu)$, we have 
\begin{equation*}
H^{-1}(\rho, \mu)\leq \frac{1}{\kappa}\sqrt{\mathcal{I}_0(\rho\|\mu)}.
\end{equation*}
\end{remark}
In next sections, we apply geometric tools in probability space to prove the above inequalities and integral formulas. 

\section{Generalized Density manifold}\label{section3}
In this section, we introduce the main tool to prove above results. We first review a class of Riemannian metrics in probability space, introduced by $\gamma-$Wasserstein distance. We then present its Riemannian calculus, including gradient and Hessian operators. 
By using gradient operators in this metric, we connect $\gamma$--divergence, $\gamma$--Fisher information and $\gamma$--drift diffusion process. 
%%We next derive the Hessian operators in these density manifolds. 
\subsection{Density manifold and its Riemannian calculus}
Consider the set of smooth and strictly positive densities
\begin{equation*}
\mathcal{P}=\Big\{\rho \in C^{\infty}(M)\colon \rho(x)>0,~\int\rho(x)dx=1\Big\}. 
\end{equation*}
The tangent space of $\mathcal{P}$ at $\rho\in \mathcal{P}$ is given by 
\begin{equation*}
  T_\rho\mathcal{P} = \Big\{\sigma\in C^{\infty}(M)\colon \int\sigma(x) dx=0 \Big\}.
\end{equation*}  
Consider the $\gamma$--Wasserstein metric tensor in the probability space as follows.
\begin{definition}\label{metric}
  The inner product $g_\rho\colon
  {T_\rho}\mathcal{P}\times{T_\rho}\mathcal{P}\rightarrow\mathbb{R}$ is defined as for any $\sigma_1$
  and $\sigma_2\in T_\rho\mathcal{P}$:
  \begin{equation*} g_\rho(\sigma_1,
    \sigma_2)=\int\Big(\sigma_1, (-\Delta_{\rho^\gamma})^{-1}\sigma_2\Big)dx,
  \end{equation*}
  where $\Delta_{\rho^\gamma}=\nabla\cdot(\rho^\gamma\nabla)$ is a weighted elliptic operator. In addition, denote $\Phi_1$, $\Phi_2\in C^{\infty}(M)/\mathbb{R}=T^*_\rho\mathcal{P}$, with $\sigma_i=-\Delta_{\rho^\gamma}\Phi_i$, $i=1, 2$,
  then 
     \begin{equation*}
  g_\rho(\sigma_1, \sigma_2)=\int (\nabla\Phi_1, \nabla\Phi_2)\rho^\gamma dx.
  \end{equation*}    
  \end{definition}
  \begin{remark}
  An observation is that if $\gamma=0$, the proposed $\mathcal{W}_0$ metric is the $H^{-1}$ metric \cite{O1}. If $\gamma=1$, the proposed $\mathcal{W}_1$ metric is the $L^2$-Wasserstein metric \cite{Lafferty, otto2001}.  
  \end{remark}
 
 We note that the characterization of geodesics in $(\mathcal{P}, g)$ has been studied in \cite{C1,O1}. In this paper, we focus on the Riemannian calculus for density manifold $(\mathcal{P}, g)$, using both $(\rho, \sigma)$ in tangent bundle and $(\rho, \Phi)$ in cotangent bundle. 
\begin{proposition}
The Christoffel symbol $\Gamma_\rho\colon T_\rho\mathcal{P}\times T_\rho\mathcal{P}\rightarrow T_\rho\mathcal{P}$ in $(\mathcal{P}, g)$ satisfies
\begin{equation*}
\begin{split}
\Gamma_\rho(\sigma_1,\sigma_2)=&-\frac{\gamma}{2}\Big\{\Delta_{\rho^{\gamma-1}\sigma_1}\Delta_{\rho^\gamma}^{-1}\sigma_2+
\Delta_{\rho^{\gamma-1}\sigma_2}\Delta_{\rho^\gamma}^{-1}\sigma_1+\Delta_{\rho^\gamma}\big((\nabla\Delta_{\rho^\gamma}^{-1}\sigma_1, \nabla\Delta_{\rho^\gamma}^{-1}\sigma_2)\rho^{\gamma-1}\big)\Big\}\\
=&-\frac{\gamma}{2}\Big\{\Delta_{\rho^{\gamma-1}\Delta_{\rho^\gamma}\Phi_1}\Phi_2+\Delta_{\rho^{\gamma-1}\Delta_{\rho^\gamma}\Phi_2}\Phi_1+\Delta_{\rho^\gamma}\big((\nabla\Phi_1, \nabla\Phi_2)\rho^{\gamma-1}\big)\Big\}
\end{split}
\end{equation*}
where $\sigma_i=-\Delta_{\rho^\gamma}\Phi_i$, $i=1,2$, and 
\begin{equation*}
\Delta_{\rho^{\gamma-1}\sigma_1}\Delta_{\rho^\gamma}^{-1}\sigma_2=\Delta_{\rho^{\gamma-1}\Delta_{\rho^\gamma}\Phi_1}\Phi_2=\nabla\cdot(\rho^{\gamma-1}\nabla\cdot(\rho^\gamma\nabla\Phi_1)\nabla\Phi_2).
\end{equation*}
\end{proposition}
\begin{proof}
The proof follows the study in \cite{LiG1}. We derive the Christoffel symbol by using the Lagrangian formulation of geodesics. 
Consider the minimization of the geometric action functional in density space
\begin{equation*}
\begin{split}
\mathcal{L}(\rho_t, \partial_t\rho_t)%=&\Big\{\int_0^1 \int\frac{1}{2} \|\nabla\Phi_t\|^2\rho_t^\gamma dxdt\colon \partial_t\rho_t+\nabla\cdot(\rho_t^\gamma \nabla\Phi_t)=0\Big\}\\
=&\int_0^1\int  \frac{1}{2}(\partial_t\rho_t, (-\Delta_{\rho_t^\gamma})^{-1}\partial_t\rho_t)dxdt, 
\end{split}
\end{equation*}
where $\rho_t=\rho(t,x)$ is a density path with fixed boundary points $\rho_0$, $\rho_1$. %and the second equality holds from applying $\Phi_t=-\Delta_{\rho_t^\gamma}^{\dd}\partial_t\rho_t$.
The geodesics in $(\mathcal{P}, g)$ satisfies the Euler-Lagrange equation 
\begin{equation}\label{EL}
\frac{\partial}{\partial t}\delta_{\partial_t\rho_t}\mathcal{L}(\rho_t, \partial_t\rho_t)=\delta_{\rho_t} \mathcal{L}(\rho_t, \partial_t\rho_t)+C(t),
\end{equation}
i.e. 
\begin{equation*}
\partial_t (-\Delta_{\rho_t^\gamma}^{-1}\partial_t\rho_t)=\delta_\rho \int \frac{1}{2}(\partial_t\rho, (-\Delta_{\rho_t^\gamma})^{-1}\partial_t\rho_t)dx+C(t),
\end{equation*}
where $C(t)$ is a function depending only on $t$.
Using the fact that 
\begin{equation*}
\partial_t\Delta_{\rho_t^\gamma}^{-1}=-\Delta_{\rho_t^\gamma}^{-1}\cdot\Delta_{\gamma \rho_t^{\gamma-1}\partial_t\rho_t}\cdot\Delta_{\rho_t^\gamma}^{-1},
\end{equation*}
then equation \eqref{EL} forms 
\begin{equation*}
-\Delta_{\rho_t^\gamma}^{-1}\partial_{tt}\rho_t+ \Delta_{\rho_t^\gamma}^{-1}\Delta_{\gamma \rho_t^{\gamma-1}\partial_t\rho_t}\Delta_{\rho_t^\gamma}^{-1}\partial_t\rho_t=-\frac{1}{2}\|\nabla \Delta_{\rho_t^\gamma}^{-1}\partial_t\rho_t\|^2\gamma\rho_t^{\gamma-1}.
\end{equation*}
By timing both sides with $\Delta_{\rho_t^\gamma}$ and comparing with the geodesics equation, 
\begin{equation*}
\partial_{tt}\rho_t+\Gamma_{\rho_t}(\partial_t\rho_t, \partial_t\rho_t)=0,
\end{equation*}
we derive the Christoffel symbol. Similarly, we can formulate the Christoffel symbol (raised Christoffel symbol) in term of $\Phi$. 
\end{proof}
\begin{proposition}\label{proposition}
The geodesics equation in $(\mathcal{P}, g)$ satisfies 
\begin{equation*}
\partial_{tt}\rho_t-\gamma\Delta_{\rho_t^{\gamma-1}\partial_t\rho_t}\Delta_{\rho_t^\gamma}^{-1}\partial_t\rho_t-\frac{\gamma}{2}\Delta_{\rho_t^\gamma}\Big(\|\nabla\Delta_{\rho_t^\gamma}^{-1}\partial_t\rho_t\|^2\rho_t^{\gamma-1}\Big)=0.
\end{equation*}
Denote Legendre transform $\Phi_t=(-\Delta_{\rho_t^\gamma})^{-1}\partial_t\rho_t$, then the co-geodesics equation satisfies 
\begin{equation}\label{cogeo}
\left\{\begin{aligned}
&\partial_t\rho_t+\nabla\cdot(\rho_t^\gamma\nabla\Phi_t)=0\\
&\partial_t\Phi_t+\frac{\gamma}{2}\|\nabla\Phi_t\|^2\rho_t^{\gamma-1}=0
\end{aligned}\right.
\end{equation}
\end{proposition}
\begin{proof}
The geodesics equation follows $\partial_{tt}\rho_t+\Gamma_{\rho_t}(\partial_t\rho_t, \partial_t\rho_t)=0$. We next demonstrate the Hamiltonian formulation of geodesics flow. 
Consider the Legendre transform in $(\mathcal{P}, g)$:
\begin{equation*}
\mathcal{H}(\rho_t, \Phi_t)=\sup_{\Phi_t\in C^{\infty}(M)} \int \Phi_t \partial_t\rho_t dx-\mathcal{L}(\rho_t,\partial_t\rho_t)
\end{equation*}
Then $\Phi_t=-\Delta_{\rho_t^\gamma}^{-1}\partial_t\rho_t$, and 
\begin{equation*}
\mathcal{H}(\rho_t,\Phi_t)=\frac{1}{2}\int \Phi_t(-\Delta_{\rho_t^\gamma}\Phi_t) dx=\frac{1}{2}\int \|\nabla\Phi_t\|^2\rho_t^\gamma dx.
\end{equation*}
Then the co-geodesic flow satisfies 
\begin{equation*}
\partial_t\rho_t=\delta_{\Phi_t}\mathcal{H}(\rho_t,\Phi_t), \quad \partial_t\Phi_t=-\delta_{\rho_t}\mathcal{H}(\rho_t,\Phi_t),
\end{equation*}
which is the equation pair \eqref{cogeo}.
\end{proof}
For completeness of this paper, we also present the Lagrangian coordinates of geodesics in generalized density manifold. 
\begin{proposition}[Lagrangian coordinates]
Denote $\rho_t=X_t\#\rho^0$, where $X_t\colon M\rightarrow M$ is the diffeomorphism and $\#$ is the push-forward operator. Then the Lagrangian coordinates of geodesic equation \eqref{cogeo} satisfies 
\begin{equation*}
\frac{d^2}{dt^2}X_t+\frac{\gamma-1}{2}\nabla\|\frac{d}{dt}X_t\|^2+({\gamma-1})\frac{d}{dt}X_t\nabla\cdot(\frac{d}{dt}X_t)-\frac{({\gamma-2})({\gamma-1})}{2}\nabla\log\rho_t \|\frac{d}{dt}X_t\|^2=0.
\end{equation*}
\end{proposition}
\begin{remark}
Here we present three examples of geodesics in Lagrangian coordinates for generalized density manifold. 
\begin{itemize}
\item[(i)] If $\gamma=1$, the $\mathcal{W}_1$ geodesic equation forms $$\frac{d^2}{dt^2}X_t=0,$$ which is a well-known result in optimal transport. 
\item[(ii)] If $\gamma=2$, the $\mathcal{W}_2$ geodesic equation satisfies 
\begin{equation*}
\frac{d^2}{dt^2}X_t+\frac{1}{2}\nabla\|\frac{d}{dt}X_t\|^2+\frac{d}{dt}X_t\nabla\cdot(\frac{d}{dt}X_t)=0.
\end{equation*}
\item[(iii)] If $\gamma=0$, the $\mathcal{W}_0$, i.e. $H^{-1}$, geodesic equation satisfies 
\begin{equation*}
\frac{d^2}{dt^2}X_t-\frac{1}{2}\nabla\|\frac{d}{dt}X_t\|^2-\frac{d}{dt}X_t\nabla\cdot(\frac{d}{dt}X_t)-\nabla\log\rho_t\|\frac{d}{dt}X_t\|^2=0.
\end{equation*}
\end{itemize}
\end{remark}

\begin{proof}
Denote 
\begin{equation*}
\frac{d}{dt}X_t(t, X_t)=\rho^{\gamma-1}(t, X_t)\nabla\Phi(t, X_t),
\end{equation*}
where $(\rho_t,\Phi_t)$ satisfies \eqref{cogeo}. Then 
\begin{equation}\label{Lagran}
\frac{d^2}{dt^2}X_t(t,X_t)=\Big\{\rho^{\gamma-1}\nabla\partial_t\Phi+\rho^{\gamma-1}\nabla\nabla\Phi \frac{d}{dt}X_t+({\gamma-1})\rho^{{\gamma-2}}\nabla\Phi\frac{d}{dt}\rho \Big\}(t, X_t).
\end{equation}
Notice the fact that 
\begin{equation*}
\begin{split}
\frac{d}{dt}\rho(t,X_t)=&\Big\{\partial_t\rho+\nabla\rho\cdot\frac{d}{dt}X_t\Big\}(t,X_t)\\
=&\Big\{\partial_t\rho+\rho^{\gamma-1}\nabla\rho\cdot \nabla\Phi\Big\}(t,X_t)\\
=&\Big\{\partial_t\rho+ \nabla\cdot(\rho^\gamma\nabla\Phi)-\rho\nabla\cdot(\rho^{\gamma-1}\nabla\Phi)\Big\}(t, X_t)\\
=&-\Big\{\rho\nabla\cdot(\rho^{\gamma-1}\nabla\Phi)\Big\}(t, X_t),
\end{split}
\end{equation*}
where the equality holds since $\partial_t\rho+\nabla\cdot(\rho^\gamma\nabla\Phi)=0$. Substituting the above and $\partial_t\Phi+\frac{\gamma}{2}\|\nabla\Phi\|^2\rho^{\gamma-1}=0$ into \eqref{Lagran}, we obtain 
\begin{equation*}
\begin{split}
\frac{d^2}{dt^2}X_t(t,X_t)=&\Big\{\rho^{\gamma-1}\nabla(-\frac{\gamma}{2}\|\nabla\Phi\|^2\rho^{\gamma-1})+\rho^{\gamma-1}\nabla\nabla\Phi \rho^{\gamma-1}\nabla\Phi\\
&+({\gamma-1})\rho^{\gamma-2}(-\rho\nabla\cdot(\rho^{\gamma-1}\nabla\Phi)\rho\nabla\Phi)\Big\}(t,X_t)\\
=&\Big\{-({\gamma-1})\rho^{\gamma-1}\nabla\nabla\Phi\rho^{\gamma-1}\nabla\Phi-\frac{r}{2}\rho^{\gamma-1}\nabla\rho^{\gamma-1}\|\nabla\Phi\|^2\\
&-({\gamma-1})\rho^{\gamma-1}\nabla\Phi \nabla\cdot(\rho^{\gamma-1}\nabla\Phi)\Big\}(t,X_t)\\
=&\Big\{-\frac{\gamma-1}{2}\nabla\|\rho^{\gamma-1}\nabla\Phi\|^2+(\gamma-1-\frac{\gamma}{2})\rho^{\gamma-1}\nabla\rho^{\gamma-1}\|\nabla\Phi\|^2\\
&-({\gamma-1})\rho^{\gamma-1}\nabla\Phi \nabla\cdot(\rho^{\gamma-1}\nabla\Phi)\Big\}(t,X_t)\\
=&\Big\{-\frac{\gamma-1}{2}\nabla\|\rho^{\gamma-1}\nabla\Phi\|^2+(\gamma-1-\frac{\gamma}{2})(\gamma-1)\nabla\log\rho^{\gamma-1}\|\rho^{\gamma-1}\nabla\Phi\|^2\\
&-({\gamma-1})\rho^{\gamma-1}\nabla\Phi \nabla\cdot(\rho^{\gamma-1}\nabla\Phi)\Big\}(t,X_t)\\
=&\Big\{-\frac{\gamma-1}{2}\nabla\|\frac{dX_t}{dt}\|^2-(\gamma-1)\frac{d}{dt}X_t\nabla\cdot(\frac{d}{dt}X_t)+\frac{(\gamma-2)(\gamma-1)}{2}\nabla\log\rho\|\frac{d}{dt}X_t\|^2\Big\}(t,X_t),
\end{split}
\end{equation*}
where the second last equality uses the fact $(\gamma-1)\log\rho \rho^{\gamma-1}=\nabla\rho^{\gamma-1}$ and $\frac{d}{dt}X_t=\rho^{\gamma-1}\nabla\Phi$. 
\end{proof}

\begin{proposition}\label{prop6}
Consider a functional $\mathcal{F}\colon \mathcal{P}\rightarrow\mathbb{R}$. 
\begin{itemize}
\item[(i)] The Riemannian gradient operator of $\mathcal{F}$ in $(\mathcal{P}, g)$ satisfies 
\begin{equation*}
\textrm{grad}_g\mathcal{F}(\rho)=-\nabla\cdot(\rho^\gamma\nabla \delta\mathcal{F}(\rho)),
\end{equation*}
where $\delta$ is the $L^2$ first variation. The squared norm of gradient operator forms 
\begin{equation*}
g_\rho(\textrm{grad}_g\mathcal{F}(\rho), \textrm{grad}_g\mathcal{F}(\rho))=\int \|\nabla\delta\mathcal{F}(\rho)\|^2\rho^\gamma dx.
\end{equation*}
\item[(ii)]
The Riemannian Hessian operator of $\mathcal{F}$ in $(\mathcal{P}, g)$ satisfies 
\begin{equation}\label{Riem}
\begin{split}
&\textrm{Hess}_g\mathcal{F}(\rho)(\sigma_1, \sigma_2)\\
=&\int \int \nabla_x\nabla_y\delta^2\mathcal{F}(\rho)(x,y)\nabla\Phi_1(x)\nabla\Phi_2(y)\rho(x)^\gamma\rho(y)^\gamma dxdy\\
&+\frac{\gamma}{2}\int \textrm{Hess} \delta\mathcal{F}(\rho)(x)(\nabla\Phi_1(x), \nabla\Phi_2(x))\rho(x)^{2\gamma-1} dx\\
&+\frac{\gamma(\gamma-1)}{2}\int \Big\{(\nabla\delta\mathcal{F}(\rho)(x), \nabla\Phi_1(x))(\nabla\Phi_2(x),\frac{\nabla\rho(x)}{\rho(x)})\\
&\hspace{2.3cm}+(\nabla\delta\mathcal{F}(\rho)(x), \nabla\Phi_2(x))(\nabla\Phi_1(x),\frac{\nabla\rho(x)}{\rho(x)})\\
&\hspace{2.3cm}-(\nabla\delta\mathcal{F}(\rho)(x), \frac{\nabla\rho(x)}{\rho(x)})(\nabla\Phi_1(x), \nabla\Phi_2(x))\Big\}\rho(x)^{2\gamma-1} dx,
\end{split}
\end{equation}
where $\sigma_i=-\Delta_{\rho^\gamma}\Phi_i$, $i=1,2$, and $\delta^2$ is the $L^2$ second variation operator. 
\end{itemize}
\end{proposition}
\begin{remark}
Several interesting examples of Hessian operators of $\mathcal{F}$ have been studied in \cite{C1}, including linear and interaction potential energies. 
\end{remark}
\begin{proof}
(i) The Riemannian gradient operator satisfies 
\begin{equation*}
g_\rho(\textrm{grad}_g\mathcal{F}(\rho), \sigma)=\int \delta\mathcal{F}(\rho)\sigma dx,\quad \textrm{for any $\sigma\in T_\rho\mathcal{P}$}.
\end{equation*}
Then 
\begin{equation*}
\begin{split}
\textrm{grad}_g\mathcal{F}(\rho)=&\Big((-\Delta_{\rho^\gamma})^{-1}\Big)^{-1}\delta\mathcal{F}(\rho)=-\Delta_{\rho^\gamma}\delta\mathcal{F}(\rho)\\
=&-\nabla\cdot(\rho^\gamma \nabla\delta\mathcal{F}(\rho)).
\end{split}
\end{equation*}

(ii) The Riemannian Hessian operator satisfies 
\begin{equation*}
\begin{split}
&\textrm{Hess}_g\mathcal{F}(\rho)(\sigma_1, \sigma_2)\\=&\int \int \delta^2\mathcal{F}(\rho)(x,y)\sigma_1(x)\sigma_2(y)dxdy-\int \delta\mathcal{F}(\rho)(x)\Gamma_\rho(\sigma_1,\sigma_2)(x)dx\\
=&\int \int \delta^2\mathcal{F}(\rho)(x,y)\nabla_x\cdot(\rho(x)^\gamma\nabla_x\Phi_1(x))\nabla_y\cdot(\rho(y)^\gamma\nabla_y\Phi_2(y))dxdy\hspace{3.2cm} (h1)\\
&+\frac{\gamma}{2}\int \delta\mathcal{F}(\rho)(x) \Big\{\Delta_{\rho^{\gamma-1}\Delta_{\rho^\gamma}\Phi_1}\Phi_2+\Delta_{\rho^{\gamma-1}\Delta_{\rho^\gamma}\Phi_2}\Phi_1+\Delta_{\rho^\gamma}\big((\nabla\Phi_1, \nabla\Phi_2)\rho^{\gamma-1}\big)\Big\} dx\qquad (h2)
\end{split}
\end{equation*}
We next formulate the terms $h_1$, $h_2$ separately. Notice the fact that 
\begin{equation*}
\begin{split}
(h1)=&\int \delta^2\mathcal{F}(\rho)(x,y)\nabla\cdot(\rho(x)^\gamma\nabla\Phi_1(x))\nabla\cdot(\rho(y)^\gamma\nabla_y\Phi_2(y))dxdy \\
=&\int \int -\nabla_x\delta^2\mathcal{F}(\rho)(x,y)\nabla\Phi_1(x)\rho(x)^{\gamma}dx \nabla_y\cdot(\rho(y)^\gamma \nabla_y\Phi_2(y))dy\\
=&\int \int \nabla_x\nabla_y\delta^2\mathcal{F}(\rho)(x,y)\nabla\Phi_1(x)\nabla\Phi_2(y)\rho(x)^\gamma\rho(y)^\gamma dxdy,
\end{split}
\end{equation*}
where the second and third equalities are shown by integration by parts with respect to $x$, $y$. In addition, we estimate three terms in (h2). 
\begin{equation*}
\begin{split}
 &\int \delta\mathcal{F}(\rho)\Delta_{\rho^{\gamma-1}\Delta_{\rho^\gamma}\Phi_1}\Phi_2 dx\\
 =&\int \delta\mathcal{F}(\rho)\nabla\cdot(\rho^{\gamma-1}\nabla\cdot(\rho^\gamma\nabla\Phi_1)\nabla\Phi_2)dx\\ 
=&-\int (\nabla\delta\mathcal{F}(\rho), \nabla\Phi_2)\nabla\cdot(\rho^\gamma\nabla\Phi_1)\rho^{\gamma-1}dx\\ 
=&\int \Big(\nabla\big((\nabla\delta\mathcal{F}(\rho), \nabla\Phi_2)\rho^{\gamma-1}\big), \nabla\Phi_1\Big)\rho^\gamma dx\\
=&\int \Big\{\textrm{Hess}\delta\mathcal{F}(\rho)(\nabla\Phi_1, \nabla\Phi_2)+\textrm{Hess}\Phi_2(\nabla\Phi_1, \nabla\delta\mathcal{F}(\rho))\Big\}\rho^{2\gamma-1}dx\\
&+\int (\nabla\delta\mathcal{F}(\rho), \nabla\Phi_2)(\nabla\rho^{\gamma-1}, \nabla\Phi_1)\rho^\gamma dx,
\end{split}
\end{equation*}
where the first and second equality holds by integration by parts with respect to $x$. Similarly, 
\begin{equation*}
\begin{split}
 &\int \delta\mathcal{F}(\rho)(x)\Delta_{\rho^{\gamma-1}\Delta_{\rho^\gamma}\Phi_2}\Phi_1 dx\\
 =&\int \Big(\nabla\big((\nabla\delta\mathcal{F}(\rho), \nabla\Phi_1)\rho^{\gamma-1}\big), \nabla\Phi_2\Big)\rho^\gamma dx\\
 =&\int \Big\{\textrm{Hess}\delta\mathcal{F}(\rho)(\nabla\Phi_1, \nabla\Phi_2)+\textrm{Hess}\Phi_1(\nabla\Phi_2, \nabla\delta\mathcal{F}(\rho))\Big\}\rho^{2\gamma-1}dx\\
&+\int (\nabla\delta\mathcal{F}(\rho), \nabla\Phi_1)(\nabla\rho^{\gamma-1}, \nabla\Phi_2)\rho^\gamma dx.
\end{split}
\end{equation*}
And 
\begin{equation*}
\begin{split}
&\int \delta\mathcal{F}(\rho)\Delta_{\rho^\gamma}(\nabla\Phi_1, \nabla\Phi_2)\rho^{\gamma-1}dx\\
=&\int \delta\mathcal{F}(\rho)\nabla\cdot(\rho^\gamma \nabla \big((\nabla\Phi_1,\nabla\Phi_2)\rho^{\gamma-1}\big)) dx \\
=&-\int \Big(\nabla\delta\mathcal{F}(\rho),  \nabla\big((\nabla\Phi_1,\nabla\Phi_2)\rho^{\gamma-1}\big)\Big)\rho^\gamma dx\\
=&-\int (\nabla\delta\mathcal{F}(\rho),  \nabla\rho^{\gamma-1})(\nabla\Phi_1, \nabla\Phi_2)\rho^\gamma dx\\
&-\int \Big\{\textrm{Hess}\Phi_1(\nabla\delta\mathcal{F}(\rho), \nabla\Phi_2)\rho^{2\gamma-1}+\textrm{Hess}\Phi_2(\nabla\delta\mathcal{F}(\rho), \nabla\Phi_1)\rho^{2\gamma-1}\Big\}dx.
\end{split}
\end{equation*}
Combining the above three terms in (h2), we finish the proof. 
\end{proof}

\subsection{Gradient systems and $\gamma$--drift diffusion process}
In this sequel, we present the relation among Riemannian gradient operators in $(\mathcal{P}, g)$, $\gamma$--divergence functional, $\gamma$--Fisher information and $\gamma$--drift diffusion process, see details in \cite{Li}.   

Given a $\gamma$-divergence functional, the Kolomogrov forward operator of $\gamma$--drift diffusion process is the negative gradient descent direction in $(\mathcal{P}, g)$. And the squared gradient norm of $\gamma-$ divergence functional in $(\mathcal{P}, g)$ forms the $\gamma$--Fisher information functional. 
\begin{lemma}
The following statements hold.
\begin{itemize}
\item[(i)]\begin{equation*}
L^*_\gamma\rho=-\textrm{grad}_g\mathcal{D}_\gamma(\rho\|\mu), 
\end{equation*}
where $L^*_\gamma$ is the adjoint operator of $L_\gamma$ in $L^2(\rho)$. 
\item[(ii)] 
\begin{equation*}
\mathcal{I}_\gamma(\rho\|\mu)=g_\rho(\textrm{grad}_g\mathcal{D}_\gamma(\rho\|\mu), \textrm{grad}_g\mathcal{D}_\gamma(\rho\|\mu)).
\end{equation*}
\end{itemize}
\end{lemma}
\begin{proof}
We first prove (i). On the one hand, the the Kolomogrov forward operator forms 
\begin{equation*}
L^*_\gamma \rho= \nabla\cdot(\mu^\gamma \nabla \frac{\rho}{\mu}).
\end{equation*}
We need to show 
\begin{equation*}
 \int L_\gamma\Phi(x) \rho(x) dx=\int \Phi(x) L^*_\gamma\rho(x) dx.
\end{equation*}
Notice the fact 
\begin{equation*}
\begin{split}
\int \rho L_\gamma \Phi(x)
=&\int \rho \Big\{(\nabla\Phi, \nabla\mu^{\gamma-1})+\mu^{\gamma-1}\Delta\Phi -\frac{1}{\gamma-1}(\nabla\Phi, \nabla\mu^{\gamma-1})\Big\}dx\\
=&\int \rho \Big\{\nabla\cdot(\mu^{\gamma-1}\nabla\Phi)-(\nabla\Phi, \nabla\mu)\mu^{\gamma-2}\Big\}dx\\
=&\int -(\nabla\Phi, \nabla \rho)\mu^{\gamma-1}-(\nabla\Phi, \nabla\mu)\mu^{\gamma-2}\rho dx\\
=&-\int (\nabla\Phi, \nabla\frac{\rho}{\mu})\mu^\gamma dx\\
=&\int \Phi\nabla\cdot(\mu^\gamma \nabla \frac{\rho}{\mu})dx\\
=&\int \Phi(x)L^*_\gamma\rho(x)dx,
\end{split}
\end{equation*}
where the second equality uses the fact $\nabla\cdot(\mu^{\gamma-1}\nabla\Phi)=(\nabla\mu^{\gamma-1},\nabla\Phi)+\mu^{\gamma-1}\Delta\Phi$ and the fourth equality applies the fact $\nabla\frac{\rho}{\mu}=\mu^{-1}\nabla\rho-\mu^{-2}\rho\nabla\mu$.
On the other hand, the negative gradient operator of $\mathcal{D}_\gamma(\rho\|\mu)$ in $(\mathcal{P}, g)$ satisfies 
\begin{equation*}
\begin{split}
-\textrm{grad}_g\mathcal{D}_\gamma(\rho\|\mu)=&\nabla\cdot(\rho^\gamma \nabla \delta\mathcal{D}_\gamma(\rho\|\mu))\\
=&\nabla\cdot(\rho^\gamma \nabla \frac{1}{1-\gamma}(\frac{\rho}{\mu})^{1-\gamma})\\
=&\nabla\cdot(\rho^\gamma (\frac{\rho}{\mu})^{-\gamma} \nabla\frac{\rho}{\mu})\\
=&\nabla\cdot(\mu^\gamma \nabla\frac{\rho}{\mu}).
\end{split}
\end{equation*}
Comparing the above, we finish the proof. 

We next prove (ii). Notice the fact that
\begin{equation*}
\begin{split}
g_\rho(\textrm{grad}_g\mathcal{D}_\gamma(\rho\|\mu), \textrm{grad}_g\mathcal{D}_\gamma(\rho\|\mu))=&\int \|\nabla \delta\mathcal{D}_\gamma(\rho\|\mu)\|^2\rho^\gamma dx\\
=&\int\|\nabla \frac{1}{1-\gamma}(\frac{\rho}{\mu})^{1-\gamma}\|^2\rho^\gamma dx\\
=&\int \|\nabla \log \frac{\rho}{\mu}\|^2(\frac{\rho}{\mu})^{2-2\gamma}\rho^\gamma dx\\
=&\int \|\nabla \log \frac{\rho}{\mu}\|^2\rho^{2-\gamma}\mu^{2-2\gamma} dx\\
=&\mathcal{I}_\gamma(\rho),
\end{split}
\end{equation*}
where the second equality uses the fact that $\frac{1}{1-\gamma}\nabla (\frac{\rho}{\mu})^{1-\gamma}=(\frac{\rho}{\mu})^{-\gamma}\nabla\frac{\rho}{\mu}=(\frac{\rho}{\mu})^{1-\gamma}\nabla\log\frac{\rho}{\mu}$. 

\end{proof}

Shortly, we shall apply the above two geometric relations to give sufficient conditions for hypercontractivity of $\gamma$--diffusion process and generalized functional inequalities. 
\section{Proof}\label{section4}

Here we present the proof of generalized diffusion hypercontractivity using the geometric tools built in previous section. 
\subsection{Sketch of proof}
%The sketch of proof is as follows.
Consider the gradient flow of $\gamma$--divergence functional 
\begin{equation*}
\partial_t\rho_t=-\textrm{grad}_g\mathcal{D}_\gamma(\rho_t\|\mu)=L^*_\gamma\rho_t,
\end{equation*}
where $\rho_t$ is the probability density function at time $t$. Then the first time derivative of $\gamma$--divergence along the gradient flow forms 
\begin{equation*}
\frac{d}{dt}\mathcal{D}_\gamma(\rho_t\|\mu)=-g_\rho(\textrm{grad}_g\mathcal{D}_\gamma(\rho_t\|\mu), \textrm{grad}_g\mathcal{D}_\gamma(\rho_t\|\mu)).
\end{equation*}
And the second time derivative of $\gamma$--divergence becomes 
\begin{equation*}
\begin{split}
\frac{d^2}{dt^2}\mathcal{D}_\gamma(\rho_t\|\mu)=&2\textrm{Hess}_g\mathcal{D}_\gamma(\rho_t\|\mu)(\partial_t\rho_t, \partial_t\rho_t)\\
=&2\textrm{Hess}_g\mathcal{D}_\gamma(\rho_t\|\mu)(\textrm{grad}_g\mathcal{D}_\gamma(\rho_t\|\mu), \textrm{grad}_g\mathcal{D}_\gamma(\rho_t\|\mu)).
\end{split}
\end{equation*}
If we can bound the ratio between the first and second derivative, i.e. 
\begin{equation}\label{compare}
\frac{d^2}{dt^2}\mathcal{D}_\gamma(\rho_t\|\mu)\geq -2 \kappa \frac{d}{dt}\mathcal{D}_\gamma(\rho_t\|\mu),
\end{equation}
we prove Theorem \ref{thm}. This is true if we integrate \eqref{compare} on both sides for $[t, \infty)$, then  
\begin{equation}\label{compare2}
-\frac{d}{dt}\mathcal{D}_\gamma(\rho_t\|\mu) \geq 2\kappa \mathcal{D}_\gamma(\rho_t\|\mu).
\end{equation}
By Grownwall's inequality, we obtain the hypercontractivity of $\gamma$--drift diffusion process
\begin{equation*}
\mathcal{D}_\gamma(\rho_t\|\mu)\leq e^{-2\kappa t}\mathcal{D}_\gamma(\rho_0\|\mu).
\end{equation*}
In addition, notice the fact that $\frac{d}{dt}\mathcal{D}_\gamma(\rho_t\|\mu)=-\mathcal{I}_\gamma(\rho_t\|\mu_t)$, then inequality \eqref{compare2} forms 
\begin{equation*}
\mathcal{I}_\gamma(\rho_t\|\mu) \geq 2\kappa \mathcal{D}_\gamma(\rho_t\|\mu).
\end{equation*}
By choosing $t=0$ with arbitrary $\rho_0\in\mathcal{P}$, the Log-Sobolev inequality \eqref{inequality} is proven. 

From above arguments, the proof boils down to estimate the ratio in \eqref{compare}. Here the formulation \eqref{compare} is equivalent to %estimate a constant $\kappa\geq 0$, such that
\begin{equation*}
\textrm{Hess}_g\mathcal{D}_\gamma(\rho\|\mu)(\textrm{grad}_g\mathcal{D}_\gamma(\rho\|\mu), \textrm{grad}_g\mathcal{D}_\gamma(\rho\|\mu))\geq \kappa g_\rho(\textrm{grad}_g\mathcal{D}_\gamma(\rho\|\mu), \textrm{grad}_g\mathcal{D}_\gamma(\rho\|\mu)).
\end{equation*}

Next, our goal is to derive the Hessian operators of $\gamma$--divergence in $(\mathcal{P}, g)$. 
\subsection{Hessian operator estimation}
%We next derive the Hessian operator of divergence function in $(\mathcal{P}, g)$. 
\begin{lemma}[Hessian of $\gamma$--divergence in $(\mathcal{P}, g)$]\label{lemma7}
Denote $\sigma=-\Delta_{\rho^\gamma}\Phi$, then
\begin{equation*}
\begin{split}
&\textrm{Hess}_g\mathcal{D}_\gamma(\rho\|\mu)(\sigma, \sigma)\\
=&\int \rho^\gamma \Big\{\big(\mu^{\gamma-1}\textrm{Ric}-\Delta\mu^{\gamma-1}-\frac{1}{\gamma-1}\textrm{Hess}\mu^{\gamma-1}\big)(\nabla\Phi, \nabla\Phi)+\mu^{\gamma-1}\|\textrm{Hess}\Phi\|^2\\
&\hspace{1.2cm}+\gamma(\gamma-1)\mu^{\gamma-1}\Big((\frac{\nabla\rho}{\rho},\nabla\Phi) (\frac{\nabla\mu}{\mu}, \nabla\Phi)-\frac{1}{2}(\frac{\nabla\rho}{\rho}, \frac{\nabla\mu}{\mu}+\frac{\nabla\rho}{\rho})(\nabla\Phi, \nabla\Phi)\Big)\Big\}dx.
\end{split}
\end{equation*}
\end{lemma}
\begin{remark}
In fact, there are several interesting cases for Hessian operators of $\gamma$--divergence in density manifold. Denote $\sigma=-\Delta_{\rho^\gamma}\Phi$.
\begin{itemize}
\item[(i)]If $\gamma=1$, then
\begin{equation*}
\begin{split}
\textrm{Hess}_g\mathcal{D}_\gamma(\rho\|\mu)(\sigma,\sigma)
=\int \rho \Big\{(\textrm{Ric}-\textrm{Hess}\log\mu)(\nabla\Phi, \nabla\Phi)+\|\textrm{Hess}\Phi\|^2\Big\}dx.
\end{split}
\end{equation*}
\item[(ii)] If $\mu=1$ is a uniform measure  \cite{C1}, then 
\begin{equation*}
\begin{split}
\textrm{Hess}_g\mathcal{D}_\gamma(\rho\|\mu)(\sigma,\sigma)
=&\int \rho^{\gamma}\Big\{\Big(\textrm{Ric}-\frac{1}{2}\gamma(\gamma-1) (\frac{\nabla\rho}{\rho},\frac{\nabla\rho}{\rho})\Big)(\nabla\Phi, \nabla\Phi)+\|\textrm{Hess}\Phi\|^2\Big\}dx.
\end{split}
\end{equation*}
\item[(iii)] If $\gamma=0$, then
\begin{equation*}\label{WH}
\begin{split}
\textrm{Hess}_g\mathcal{D}_\gamma(\rho\|\mu)(\sigma, \sigma)
%=&\int \mu^{-1}(\Delta\Phi)^2dx\\
=&\int \Big\{\big(\mu^{-1}\textrm{Ric}-\Delta\mu^{-1}+\textrm{Hess}\mu^{-1}\big)(\nabla\Phi, \nabla\Phi)+\mu^{-1}\|\textrm{Hess}\Phi\|^2\Big\} dx.
\end{split}
\end{equation*}
%where the last equality is obtained by computing the Hessian of Pearson divergence in $H^{-1}$ metric by using tangent bundle directly. 
\end{itemize}
\end{remark}

\begin{proof}
From proposition \ref{prop6}, we can compute the Hessian operator of $\gamma$--Divergence directly. For readers who are not family with geometric computations, the following direct method is also given. 
The Hessian operator is given by taking the second order time derivative of $\mathcal{D}_\gamma(\rho\|\mu)$ along the co-geodesics flow \eqref{cogeo}. 
Consider the first order time derivative
\begin{equation*}
\begin{split}
\frac{d}{dt}\mathcal{D}_\gamma(\rho_t\|\mu)=&\int \delta \mathcal{D}_\gamma(\rho_t\|\mu)\partial_t\rho_t dx\\
=&\int \delta \mathcal{D}_\gamma(\rho_t\|\mu) \Big(-\nabla\cdot(\rho_t^\gamma\nabla\Phi_t)\Big)dx\\
=&\int (\nabla\delta \mathcal{D}_\gamma(\rho_t\|\mu), \nabla\Phi_t) \rho_t^\gamma dx.
\end{split}
\end{equation*}
And the second order time derivative satisfies
\begin{equation*}
\begin{split}
\textrm{Hess}_g\mathcal{D}_\gamma(\rho\|\mu)(\sigma, \sigma)=&\frac{d^2}{dt^2}\mathcal{D}_\gamma(\rho_t\|\mu)|_{t=0}\\
=&\frac{d}{dt}\Big(\frac{d}{dt}\mathcal{D}_\gamma(\rho_t\|\mu)\Big)|_{t=0}\\
=&\int (\nabla\frac{d}{dt}\delta \mathcal{D}_\gamma(\rho_t\|\mu), \nabla\Phi_t)\rho_t^\gamma dx|_{t=0}\hspace{3.5cm}(a)\\
&+\int(\nabla\delta \mathcal{D}_\gamma(\rho_t\|\mu), \nabla \partial_t\Phi_t)\rho_t^\gamma dx|_{t=0} \hspace{3.2cm}(b)\\
&+\int (\nabla\delta \mathcal{D}_\gamma(\rho_t\|\mu), \nabla\Phi_t)(\gamma \rho_t^{\gamma-1})\partial_t\rho_t dx|_{t=0} \hspace{2cm}(c)
\end{split}
\end{equation*}
We estimate (a), (b) and (c) separately. 

For (a), we denote $\delta^2\mathcal{D}(\rho)=\frac{\partial^2}{\partial\rho\partial\rho}(f(\frac{\rho}{\mu})\mu)(x)$. Then
\begin{equation*}
\begin{split}
(a)=&\int \delta^2\mathcal{D}_\gamma(\rho\|\mu) \Big(\nabla\cdot (\rho^\gamma\nabla\Phi)\Big)^2 dx\\
=&\int\delta^2 \mathcal{D}_\gamma(\rho\|\mu)\Big((\nabla \rho^\gamma, \nabla\Phi)+\rho^\gamma \Delta\Phi\Big)^2dx\\
=&\int \delta^2 \mathcal{D}_\gamma(\rho\|\mu)\Big( (\nabla \rho^\gamma , \nabla\Phi)^2+2(\nabla \rho^\gamma, \nabla\Phi)\rho^\gamma\Delta\Phi+\rho^{2\gamma}(\Delta\Phi)^2 \Big)\\
=&\int \rho^{-\gamma}\mu^{\gamma-1}\Big((\nabla \rho^\gamma, \nabla\Phi)^2+2(\nabla \rho^\gamma, \nabla\Phi)\rho^\gamma\Delta\Phi+\rho^{2\gamma}(\Delta\Phi)^2\Big) dx\\
=&\int \rho^{-\gamma}\mu^{\gamma-1}\gamma^2\rho^{2\gamma-2}(\nabla \rho, \nabla\Phi)^2 dx+\int 2\rho^{-\gamma}\mu^{\gamma-1}(\nabla \rho^\gamma, \nabla\Phi)\rho^\gamma\Delta\Phi dx\\
&+\int \mu^{\gamma-1} \rho^{\gamma}(\Delta\Phi)^2 dx\\
=&\int \mu^{\gamma-1}\gamma^2\rho^{\gamma-2}(\nabla \rho, \nabla\Phi)^2 dx+\int 2\mu^{\gamma-1}(\nabla \rho^\gamma, \nabla\Phi)\Delta\Phi dx\\
&+\int \mu^{\gamma-1} \rho^{\gamma}(\Delta\Phi)^2 dx\\
=&\gamma^2\int \mu^{\gamma-1}\rho^{\gamma}(\frac{1}{\rho}\nabla \rho, \nabla\Phi)^2 dx\\
&-2\int \rho^\gamma \nabla\cdot(\mu^{\gamma-1}\nabla\Phi\Delta\Phi) dx\\
&+\int \mu^{\gamma-1} \rho^{\gamma}(\Delta\Phi)^2 dx.
%=&\gamma^2\int \mu^{\gamma-1}\rho^{\gamma}(\frac{1}{\rho}\nabla \rho, \nabla\Phi)^2 dx\\
%&-2\int \rho^\gamma \Big\{(\nabla\mu^{\gamma-1},\nabla\Phi)\Delta\Phi+\mu^{\gamma-1}(\Delta\Phi)^2+\mu^{\gamma-1}(\nabla\Delta\Phi, \nabla\Phi) \Big\} dx\\
%&+\int \mu^{\gamma-1} \rho^{\gamma}(\Delta\Phi)^2 dx.
\end{split}
\end{equation*}
We next estimate (b).
\begin{equation*}
\begin{split}
(b)=&\int(\nabla\delta \mathcal{D}_\gamma(\rho_t\|\mu), \nabla \partial_t\Phi_t)\rho_t^\gamma dx|_{t=0}\\
=&\int \partial_t\Phi_t \Big(-\nabla\cdot(\rho_t^\gamma\nabla\delta \mathcal{D}_\gamma(\rho_t\|\mu))\Big)dx|_{t=0}\\
=&\int \frac{1}{2}(\nabla\Phi, \nabla\Phi)\gamma \rho^{\gamma-1} \nabla\cdot(\rho^\gamma\nabla\delta \mathcal{D}_\gamma(\rho\|\mu)) dx\\
=&\frac{1}{2}\int (\nabla\Phi, \nabla\Phi)\gamma \rho^{\gamma-1} \nabla\cdot(\rho^\gamma \nabla\frac{1}{1-\gamma}(\frac{\rho}{\mu})^{1-\gamma}) dx\\
=&\frac{1}{2}\int (\nabla\Phi, \nabla\Phi)\gamma\rho^{\gamma-1} \nabla\cdot(\rho^\gamma (\frac{\rho}{\mu})^{-\gamma} \nabla\frac{\rho}{\mu}) dx\\
=&\frac{1}{2}\int (\nabla\Phi, \nabla\Phi)\gamma\rho^{\gamma-1} \nabla\cdot( \mu^\gamma \frac{\nabla\rho\mu-\rho\nabla\mu}{\mu^2}) dx\\
=&\frac{1}{2}\int (\nabla\Phi, \nabla\Phi)\gamma\rho^{\gamma-1} \nabla\cdot( \mu^{\gamma-1}\nabla\rho - \rho\mu^{\gamma-2}\nabla\mu) dx\\
=&-\frac{1}{2}\int \nabla\Big((\nabla\Phi, \nabla\Phi)\gamma\rho^{\gamma-1}\Big)\Big(\mu^{\gamma-1}\nabla\rho - \frac{1}{\gamma-1}\rho\nabla\mu^{\gamma-1}\Big) dx\\
=&-\frac{1}{2}\int \Big(\nabla(\nabla\Phi, \nabla\Phi)\gamma\rho^{\gamma-1}+(\nabla\Phi, \nabla\Phi)\gamma\nabla\rho^{\gamma-1} \Big)\Big(\mu^{\gamma-1}\nabla\rho - \frac{1}{\gamma-1}\rho\nabla\mu^{\gamma-1}\Big) dx\\
=&-\frac{1}{2}\int (\nabla(\nabla\Phi, \nabla\Phi), \nabla\rho)\gamma\rho^{\gamma-1}\mu^{\gamma-1} dx\hspace{4.1cm}(b1)\\
&-\frac{1}{2}\int (\nabla\rho^{\gamma-1},\nabla\rho)(\nabla\Phi, \nabla\Phi)\gamma\mu^{\gamma-1} dx\hspace{4.1cm} (b2)\\
&+\frac{1}{2(\gamma-1)}\int\Big(\nabla(\nabla\Phi, \nabla\Phi)\gamma\rho^{\gamma},\nabla\mu^{\gamma-1}\Big)dx\hspace{3.2cm}(b3)\\
&+ \frac{1}{2(\gamma-1)}\int (\nabla\Phi, \nabla\Phi)\gamma(\nabla\rho^{\gamma-1}\rho,\nabla\mu^{\gamma-1})dx\hspace{2.8cm}(b4)\\
\end{split}
\end{equation*}
Here we derive (b1), (b2), (b3), (b4) more explicitly. Notice the fact that 
\begin{equation*}
\begin{split}
(b1)=&-\frac{1}{2}\int (\nabla(\nabla\Phi, \nabla\Phi), \nabla\rho)\gamma\rho^{\gamma-1}\mu^{\gamma-1} dx\\
=&-\frac{1}{2}\int \nabla(\nabla\Phi, \nabla\Phi)\mu^{\gamma-1}\nabla\rho^\gamma dx\\
=&\frac{1}{2}\int \nabla\cdot \Big(\mu^{\gamma-1}\nabla(\nabla\Phi, \nabla\Phi)\Big)\rho^\gamma dx\\
=&\frac{1}{2}\int \Big\{(\nabla\mu^{\gamma-1},\nabla(\nabla\Phi, \nabla\Phi))+\mu^{\gamma-1}\Delta(\nabla\Phi, \nabla\Phi)\Big\}\rho^\gamma dx,
\end{split}
\end{equation*}
and
\begin{equation*}
\begin{split}
(b2)=&-\frac{1}{2}\int (\nabla\Phi, \nabla\Phi)(\gamma\nabla\rho^{\gamma-1},\mu^{\gamma-1}\nabla\rho) dx\\
=&-\frac{1}{2}\gamma(\gamma-1)\int (\nabla\Phi, \nabla\Phi)\rho^{\gamma-2}\mu^{\gamma-1}(\nabla\rho,\nabla\rho) dx\\
=&-\frac{1}{2}\gamma(\gamma-1)\int \rho^{\gamma}\mu^{\gamma-1} (\nabla\Phi, \nabla\Phi)(\frac{1}{\rho}\nabla\rho,\frac{1}{\rho}\nabla\rho) dx
\end{split}
\end{equation*}
In addition, 
\begin{equation*}
\begin{split}
(b3)=&\frac{1}{2(\gamma-1)}\int\Big(\nabla(\nabla\Phi, \nabla\Phi)\gamma\rho^{\gamma},\nabla\mu^{\gamma-1}\Big)dx\\
=&\frac{\gamma}{2(\gamma-1)}\int \rho^\gamma\Big(\nabla(\nabla\Phi, \nabla\Phi),\nabla\mu^{\gamma-1}\Big)dx
\end{split}
\end{equation*}
and 
\begin{equation*}
\begin{split}
(b4)=& \frac{1}{2(\gamma-1)}\int (\nabla\Phi, \nabla\Phi)\gamma(\nabla\rho^{\gamma-1}\rho,\nabla\mu^{\gamma-1})dx\\
=&\frac{1}{2(\gamma-1)}\int (\nabla\Phi, \nabla\Phi)(\gamma (\gamma-1)\rho^{\gamma-2}\rho\nabla\rho,\nabla\mu^{\gamma-1})dx\\
=&\frac{1}{2}\int (\nabla\Phi, \nabla\Phi)(\nabla\rho^{\gamma},\nabla\mu^{\gamma-1})dx\\
=&-\frac{1}{2}\int \rho^{\gamma}\nabla\cdot\Big((\nabla\Phi, \nabla\Phi)\nabla\mu^{\gamma-1}\Big)dx\\
=&-\frac{1}{2}\int \rho^{\gamma}\Big\{(\nabla(\nabla\Phi, \nabla\Phi),\nabla\mu^{\gamma-1})+(\nabla\Phi, \nabla\Phi) \Delta\mu^{\gamma-1}\Big\}dx.
\end{split}
\end{equation*}

We last derive (c).
\begin{equation*}
\begin{split}
(c)=&\int (\nabla\delta \mathcal{D}_\gamma(\rho_t\|\mu), \nabla\Phi_t)(\gamma \rho_t^{\gamma-1})\partial_t\rho_t dx|_{t=0}\\
=&\int  (\nabla \frac{1}{1-\gamma}(\frac{\rho}{\mu})^{1-\gamma}, \nabla\Phi)\gamma\rho^{\gamma-1} \Big(-\nabla\cdot(\rho^\gamma\nabla\Phi)\Big)dx\\
=&\int  (\frac{\rho}{\mu})^{-\gamma}(\nabla\frac{\rho}{\mu}, \nabla\Phi)\gamma\rho^{\gamma-1} \Big(-\nabla\cdot(\rho^\gamma\nabla\Phi)\Big)dx\\
=&\gamma\int (\nabla\frac{\rho}{\mu}, \nabla\Phi)\mu^\gamma \Big(-\frac{1}{\rho}\nabla\cdot(\rho^\gamma\nabla\Phi)\Big)dx\\
=&\gamma\int (\frac{\nabla\rho}{\mu}-\frac{\rho\nabla\mu}{\mu^2}, \nabla\Phi)\mu^\gamma \Big(-\frac{1}{\rho}\nabla\cdot(\rho^\gamma\nabla\Phi)\Big)dx\\
=&\gamma\int (\nabla\rho \mu^{\gamma-1}-\mu^{\gamma-2}{\rho\nabla\mu}, \nabla\Phi) \Big(-\frac{1}{\rho}\nabla\cdot(\rho^\gamma\nabla\Phi)\Big)dx\\
=&-\gamma\int (\nabla\rho \mu^{\gamma-1}-\mu^{\gamma-2}{\rho\nabla\mu}, \nabla\Phi) \Big((\frac{1}{\rho}\nabla\rho^\gamma, \nabla\Phi)+\rho^{\gamma-1}\Delta\Phi\Big)dx\\
=&-\gamma\int (\nabla\rho \mu^{\gamma-1},\nabla\Phi)(\frac{1}{\rho}\nabla\rho^\gamma, \nabla\Phi)dx\hspace{2.3cm} (c1)\\
&-\gamma\int (\nabla\rho ,\nabla\Phi)\rho^{\gamma-1} \mu^{\gamma-1}\Delta\Phi dx \hspace{3cm}(c2)\\
&+\gamma\int (\mu^{\gamma-2}{\nabla\mu},\nabla\Phi)(\nabla\rho^\gamma, \nabla\Phi) dx\hspace{2.5cm}(c3)\\
&+\gamma\int (\mu^{\gamma-2}{\nabla\mu}, \nabla\Phi) \rho^{\gamma}\Delta\Phi dx\hspace{3.3cm}(c4)
\end{split}
\end{equation*}
We estimate (c1), (c2), (c3), (c4) explicitly. Notice the fact that 
\begin{equation*}
\begin{split}
(c1)=&-\gamma\int (\nabla\rho \mu^{\gamma-1},\nabla\Phi)(\frac{1}{\rho}\nabla\rho^\gamma, \nabla\Phi)dx\\
=&-\gamma^2\int \mu^{\gamma-1} (\nabla\rho ,\nabla\Phi)^2\rho^{\gamma-2}dx\\
=&-\gamma^2\int \mu^{\gamma-1} (\frac{1}{\rho}\nabla\rho ,\nabla\Phi)^2\rho^{\gamma}dx,
\end{split}
\end{equation*}
and 
\begin{equation*}
\begin{split}
(c2)=&-\gamma\int (\nabla\rho ,\nabla\Phi)\rho^{\gamma-1} \mu^{\gamma-1}\Delta\Phi dx\\
=&-\int (\nabla\rho^\gamma ,\nabla\Phi) \mu^{\gamma-1}\Delta\Phi dx\\
=&\int \rho^\gamma \nabla\cdot(\mu^{\gamma-1}\Delta\Phi\nabla\Phi)dx.
%=&\int \rho^\gamma\Big\{(\nabla\mu^{\gamma-1}, \nabla\Phi)\Delta\Phi+\mu^{\gamma-1}(\nabla\Delta\Phi, \nabla\Phi)+\mu^{\gamma-1}(\Delta\Phi)^2\Big\}dx.
\end{split}
\end{equation*}
In addition, 
\begin{equation*}
\begin{split}
(c3)+(c4)=&\gamma\int (\mu^{\gamma-2}{\nabla\mu}, \nabla\Phi)\nabla\cdot(\rho^\gamma\nabla\Phi) dx\\
%=&\frac{\gamma}{\gamma-1}\int (\nabla\mu^{\gamma-1},\nabla\Phi)(\nabla\rho^\gamma, \nabla\Phi)dx\\
=&-\frac{\gamma}{\gamma-1}\int \rho^\gamma \Big(\nabla(\nabla\mu^{\gamma-1},\nabla\Phi), \nabla\Phi\Big)dx.
\end{split}
\end{equation*}
We now summarize all the formulas. 
\begin{equation}\label{main}
\begin{split}
&\textrm{Hess}_g\mathcal{D}_\gamma(\rho\|\mu)(\sigma, \sigma)\\=&\frac{d^2}{dt^2}\mathcal{D}_\gamma(\rho_t\|\mu)|_{t=0}=(a)+(b)+(c)\\
=&(a)+(b1)+(b2)+(b3)+(b4)+(c1)+(c2)+(c3)+(c4)\\
=&\gamma^2\int \mu^{\gamma-1}\rho^{\gamma}(\frac{1}{\rho}\nabla \rho, \nabla\Phi)^2 dx\\
&-2\int \rho^\gamma \nabla\cdot(\mu^{\gamma-1}\nabla\Phi\Delta\Phi) dx\\
&+\int \mu^{\gamma-1} \rho^{\gamma}(\Delta\Phi)^2 dx\\
&+\frac{1}{2}\int \Big\{(\nabla\mu^{\gamma-1},\nabla(\nabla\Phi, \nabla\Phi))+\mu^{\gamma-1}\Delta(\nabla\Phi, \nabla\Phi)\Big\}\rho^\gamma dx\\
&-\frac{1}{2}\gamma(\gamma-1)\int \rho^{\gamma}\mu^{\gamma-1} (\nabla\Phi, \nabla\Phi)(\frac{1}{\rho}\nabla\rho,\frac{1}{\rho}\nabla\rho) dx\\
&+\frac{\gamma}{2(\gamma-1)}\int \rho^\gamma\Big(\nabla(\nabla\Phi, \nabla\Phi),\nabla\mu^{\gamma-1}\Big)dx\\
&-\frac{1}{2}\int \rho^{\gamma}\Big\{(\nabla(\nabla\Phi, \nabla\Phi),\nabla\mu^{\gamma-1})+(\nabla\Phi, \nabla\Phi) \Delta\mu^{\gamma-1}\Big\}dx\\
&-\gamma^2\int \mu^{\gamma-1} (\frac{1}{\rho}\nabla\rho ,\nabla\Phi)^2\rho^{\gamma}dx\\
&+\int \rho^\gamma \nabla\cdot(\mu^{\gamma-1}\Delta\Phi\nabla\Phi)dx\\
&-\frac{\gamma}{\gamma-1}\int \rho^\gamma \Big(\nabla(\nabla\mu^{\gamma-1},\nabla\Phi), \nabla\Phi\Big)dx\\
=& \int \rho^\gamma\Big\{ -\nabla\cdot(\mu^{\gamma-1}\nabla\Phi\Delta\Phi)+\mu^{\gamma-1}(\Delta\Phi)^2+\frac{1}{2}\mu^{\gamma-1}\Delta(\nabla\Phi,\nabla\Phi) \qquad (d)\\
&\qquad-\frac{1}{2}\Delta\mu^{\gamma-1}(\nabla\Phi,\nabla\Phi)-\frac{\gamma}{\gamma-1}\textrm{Hess}\mu^{\gamma-1}(\nabla\Phi, \nabla\Phi)\\
&\qquad-\frac{1}{2}\gamma(\gamma-1)(\frac{\nabla\rho}{\rho},\frac{\nabla\rho}{\rho})(\nabla\Phi, \nabla\Phi)\mu^{\gamma-1}\Big\} dx.
\end{split}
\end{equation}
Notice the fact that 
\begin{equation*}
\begin{split}
(d)=&-\nabla\cdot(\mu^{\gamma-1}\nabla\Phi\Delta\Phi)+\mu^{\gamma-1}(\Delta\Phi)^2+\frac{1}{2}\mu^{\gamma-1}\Delta(\nabla\Phi,\nabla\Phi)\\
=&-(\nabla\mu^{\gamma-1}, \nabla\Phi)\Delta\Phi-\mu^{\gamma-1}(\Delta\Phi)^2-\mu^{\gamma-1}(\nabla\Delta\Phi, \nabla\Phi)\\
&+\mu^{\gamma-1}(\Delta\Phi)^2+\frac{1}{2}\mu^{\gamma-1}\Delta(\nabla\Phi,\nabla\Phi)\\
=&-(\nabla\mu^{\gamma-1},\nabla\Phi)\Delta\Phi+\mu^{\gamma-1} \Big\{\frac{1}{2}\Delta(\nabla\Phi, \nabla\Phi)-(\nabla\Phi, \nabla\Delta\Phi) \Big\}  \\
=&-(\nabla\mu^{\gamma-1}, \nabla\Phi)\Delta\Phi+\mu^{\gamma-1} \Big\{\textrm{Ric}(\nabla\Phi,\nabla\Phi)+ \|\textrm{Hess}\Phi\|^2\Big\},
\end{split}
\end{equation*}
where the last equality is from Bochner's formula, i.e. 
\begin{equation*}
\frac{1}{2}\Delta(\nabla\Phi, \nabla\Phi)-(\nabla\Phi, \nabla\Delta\Phi)=\textrm{Ric}(\nabla\Phi,\nabla\Phi)+ \|\textrm{Hess}\Phi\|^2.
\end{equation*}

By substituting (d) into \eqref{main}, we obtain 
\begin{equation}\label{main1}
\begin{split}
&\textrm{Hess}_g\mathcal{D}_\gamma(\rho\|\mu)(V_\Phi, V_\Phi)\\
=&\int \rho^\gamma \mu^{\gamma-1}\Big\{(\textrm{Ric}(\nabla\Phi, \nabla\Phi)+\|\textrm{Hess}\Phi\|^2\Big\} dx\\
&+\frac{1}{2}\int \rho^\gamma\Delta\mu^{\gamma-1}(\nabla\Phi, \nabla\Phi) dx\\
&-\frac{r}{\gamma-1}\int \rho^\gamma \textrm{Hess}\mu^{\gamma-1}(\nabla\Phi, \nabla\Phi)dx\\
&-\frac{1}{2}\gamma(\gamma-1)\int \rho^\gamma \mu^{\gamma-1} (\frac{\nabla\rho}{\rho},\frac{\nabla\rho}{\rho})(\nabla\Phi, \nabla\Phi)dx\\
&-\int \rho^\gamma (\nabla\mu^{\gamma-1}, \nabla\Phi)\Delta\Phi dx.   \hspace{2cm}{(e)}
\end{split}
\end{equation}

We lastly reformulate the term (e): 
\begin{equation*}
\begin{split}
(e)=&-\int \rho^\gamma (\nabla\mu^{\gamma-1}, \nabla\Phi)\nabla\cdot(\nabla\Phi) dx\\
=&\int \Big(\nabla \big(\rho^\gamma (\nabla\mu^{\gamma-1}, \nabla\Phi)\big), \nabla \Phi\Big)dx\\
=&\int (\nabla\rho^\gamma,\nabla\Phi) (\nabla\mu^{\gamma-1}, \nabla\Phi)dx+\int \rho^\gamma(\nabla (\nabla\mu^{\gamma-1}, \nabla\Phi), \nabla\Phi) dx\\
=&\int (\nabla\rho^\gamma,\nabla\Phi) (\nabla\mu^{\gamma-1}, \nabla\Phi)dx+\int \rho^\gamma\textrm{Hess}\mu^{\gamma-1}(\nabla\Phi, \nabla\Phi) dx\\
&+\int \rho^\gamma \textrm{Hess}\Phi(\nabla\Phi, \nabla\mu^{\gamma-1})dx \hspace{2cm}{(e1)} 
\end{split}
\end{equation*}
Notice that (e1) has a formulation 
\begin{equation*}
\begin{split}
(e1)=&\int \rho^\gamma \textrm{Hess}\Phi(\nabla\Phi, \nabla\mu^{\gamma-1})dx\\
=&\int\rho^\gamma \Big(\nabla\big(\frac{1}{2}(\nabla\Phi)^2\big), \nabla\mu^{\gamma-1}\Big)dx\\
=&-\frac{1}{2}\int \nabla\cdot(\rho^\gamma \nabla\mu^{\gamma-1}) (\nabla\Phi)^2dx\\
=&-\frac{1}{2}\int \Big((\nabla\rho^\gamma, \nabla\mu^{\gamma-1})+\rho^\gamma\Delta\mu^{\gamma-1}\Big) (\nabla\Phi)^2dx\\
\end{split}
\end{equation*}
where the second equality holds by the fact that $\textrm{Hess}\Phi\nabla\Phi=\nabla\nabla\Phi\nabla\Phi=\frac{1}{2}\nabla (\nabla\Phi)^2$.

Substituting (e1) into (e), we obtain 
\begin{equation*}
\begin{split}
(e)=&\int (\nabla\rho^\gamma,\nabla\Phi) (\nabla\mu^{\gamma-1}, \nabla\Phi)dx-\frac{1}{2}\int (\nabla\rho^\gamma, \nabla\mu^{\gamma-1})(\nabla\Phi, \nabla\Phi)dx\\
&+\int \rho^\gamma\textrm{Hess}\mu^{\gamma-1}(\nabla\Phi, \nabla\Phi) dx-\frac{1}{2}\int \rho^\gamma\Delta\mu^{\gamma-1} (\nabla\Phi,\nabla\Phi)dx\\
=&\gamma(\gamma-1)\int \rho^\gamma\mu^{\gamma-1}\Big\{(\frac{\nabla\rho}{\rho},\nabla\Phi) (\frac{\nabla\mu}{\mu}, \nabla\Phi)-\frac{1}{2}(\frac{\nabla\rho}{\rho}, \frac{\nabla\mu}{\mu})(\nabla\Phi, \nabla\Phi)\Big\}dx\\
&+\int \rho^\gamma\textrm{Hess}\mu^{\gamma-1}(\nabla\Phi, \nabla\Phi) dx-\frac{1}{2}\int \rho^\gamma\Delta\mu^{\gamma-1} (\nabla\Phi,\nabla\Phi)dx.
\end{split}
\end{equation*}
Substituting the above formula into \eqref{main1}, we derive 
\begin{equation*}
\begin{split}
&\textrm{Hess}_g\mathcal{D}_\gamma(\rho\|\mu)(V_\Phi, V_\Phi)\\
=&\int \rho^\gamma \mu^{\gamma-1}\Big\{\textrm{Ric}(\nabla\Phi, \nabla\Phi)+\|\textrm{Hess}\Phi\|^2\Big\} dx-\frac{1}{2}\int \rho^\gamma \Delta\mu^{\gamma-1}(\nabla\Phi, \nabla\Phi)dx\\
&-\frac{r}{\gamma-1}\int \rho^\gamma \textrm{Hess}\mu^{\gamma-1}(\nabla\Phi, \nabla\Phi)dx\\
&-\frac{1}{2}\gamma(\gamma-1)\int \rho^\gamma \mu^{\gamma-1}(\frac{\nabla\rho}{\rho},\frac{\nabla\rho}{\rho})(\nabla\Phi, \nabla\Phi)dx\\
&+\gamma(\gamma-1)\int \rho^\gamma\mu^{\gamma-1}\Big\{(\frac{\nabla\rho}{\rho},\nabla\Phi) (\frac{\nabla\mu}{\mu}, \nabla\Phi)-\frac{1}{2}(\frac{\nabla\rho}{\rho}, \frac{\nabla\mu}{\mu})(\nabla\Phi, \nabla\Phi)\Big\}dx\\
&+\int \rho^\gamma\textrm{Hess}\mu^{\gamma-1}(\nabla\Phi, \nabla\Phi) dx-\frac{1}{2}\int \rho^\gamma\Delta\mu^{\gamma-1} (\nabla\Phi,\nabla\Phi)dx\\
=&\int \rho^\gamma \Big\{\big(\mu^{\gamma-1}\textrm{Ric}-\Delta\mu^{\gamma-1}-\frac{1}{\gamma-1}\textrm{Hess}\mu^{\gamma-1}\big)(\nabla\Phi, \nabla\Phi)+\mu^{\gamma-1}\|\textrm{Hess}\Phi\|^2\\
&\hspace{0.5cm}+\gamma(\gamma-1)\mu^{\gamma-1}\Big((\frac{\nabla\rho}{\rho},\nabla\Phi) (\frac{\nabla\mu}{\mu}, \nabla\Phi)-\frac{1}{2}(\frac{\nabla\rho}{\rho}, \frac{\nabla\mu}{\mu}+\frac{\nabla\rho}{\rho})(\nabla\Phi, \nabla\Phi)\Big)\Big\}dx.
\end{split}
\end{equation*}
\end{proof}

We observe that the Hessian operator in $(\mathcal{P}, g)$ is more complicated than the one with $\gamma=1$. Since when $\gamma=1$, there is no interaction bilinear term between the Hessian operator and the squared gradient norm. We overcome this by the following estimates. Denote the bilinear form:
\begin{equation*}
J(\Phi,\Phi)=(\frac{\nabla \rho}{\rho},\nabla\Phi)(\frac{\nabla\mu}{\mu},\nabla\Phi)-\frac{1}{2}(\frac{\nabla\rho}{\rho}, \frac{\nabla\rho}{\rho}+\frac{\nabla\mu}{\mu})(\nabla\Phi, \nabla\Phi).
\end{equation*}
\begin{lemma}\label{lemma8}
Denote $\delta\mathcal{D}_\gamma(\rho\|\mu)=\frac{1}{1-\gamma}(\frac{\rho}{\mu})^{1-\gamma}$, then for any $\rho\in \mathcal{P}$,
\begin{equation*}
\frac{ J(\delta\mathcal{D}_\gamma(\rho\|\mu), \delta\mathcal{D}_\gamma(\rho\|\mu))}{(\nabla \delta\mathcal{D}_\gamma(\rho\|\mu), \nabla \delta \mathcal{D}_\gamma(\rho\|\mu))}\in (-\infty, \frac{1}{8}\|\nabla\log\mu\|^2].
\end{equation*}
\end{lemma}
\begin{proof}
The proof is based on an estimation for the bilinear form $J$. Notice 
\begin{equation*}
\nabla \delta\mathcal{D}_\gamma(\rho\|\mu)=(\frac{\rho}{\mu})^{-\gamma} \nabla\frac{\rho}{\mu}=(\frac{\rho}{\mu})^{1-\gamma}\nabla\log\frac{\rho}{\mu}.
\end{equation*}
Then 
\begin{equation*}
\begin{split}
J_1:=&J(\delta \mathcal{D}_\gamma(\rho\|\mu), \delta \mathcal{D}_\gamma(\rho\|\mu))\\
=&\Big\{(\nabla\log\rho, \nabla\log\frac{\rho}{\mu})(\nabla\log\mu, \nabla\log\frac{\rho}{\mu})\\
&-\frac{1}{2}(\nabla\log\rho, \nabla\log\rho+\nabla\log\mu)(\nabla\log\frac{\rho}{\mu}, \nabla\log\frac{\rho}{\mu})\Big\}(\frac{\rho}{\mu})^{2-2\gamma},
\end{split}
\end{equation*}
and 
\begin{equation*}
J_2:=(\nabla \delta\mathcal{D}_\gamma(\rho\|\mu), \nabla \delta \mathcal{D}_\gamma(\rho\|\mu))=(\nabla\log\frac{\rho}{\mu}, \nabla\log\frac{\rho}{\mu})(\frac{\rho}{\mu})^{2-2\gamma}.
\end{equation*}
Denote $\nabla\log\frac{\rho}{\mu}=a$, $\nabla \log\mu=a_0$, then $\nabla\log\rho=a+a_0$, and thus
\begin{equation*}
\begin{split}
\frac{J_1}{J_2}=&\frac{(a+a_0, a)(a_0, a)-\frac{1}{2}(a+a_0, a+2a_0)(a, a)}{(a,a)}\\
=&\frac{(a,a)(a_0,a)+(a_0, a)^2-\frac{1}{2} [(a,a)+3(a_0, a)+2(a_0,a_0)](a,a)}{(a,a)}\\
=&\frac{(a,a)(a_0,a)+(a_0,a)^2-\frac{1}{2}(a,a)^2-\frac{3}{2}(a_0,a)(a,a)-(a_0,a_0)(a,a)}{(a,a)}\\
=&\frac{(a_0,a)^2-\frac{1}{2}(a,a)^2-\frac{1}{2}(a_0,a)(a,a)-(a_0,a_0)(a,a)}{(a,a)}.
\end{split}
\end{equation*}
We further denote $\cos\theta=\frac{(a_0,a)}{\|a\|\|a_0\|}$, then 
\begin{equation*}
\begin{split}
\frac{J_1}{J_2}=& \frac{\|a\|^2\|a_0\|^2\cos^2\theta-\frac{1}{2}\|a\|^4-\frac{1}{2}\|a_0\|\|a\|^3\cos\theta-\|a_0\|^2\|a\|^2}{\|a\|^2}\\
=&\|a_0\|^2(\cos^2\theta-1)-\frac{1}{2}\|a\|^2-\frac{1}{2}\|a\|\|a_0\|\cos\theta\\
=&\|a_0\|^2(\frac{9}{8}\cos^2\theta-1)-\frac{1}{2}\big(\|a\|+\frac{1}{2}\|a_0\|\cos\theta\big)^2\\
\leq &\frac{1}{8}\|a_0\|^2,
\end{split}
\end{equation*}
which finishes the proof.
\end{proof}
\subsection{Proof}
\begin{proof}[Proof of Theorem \ref{thm}]
Firstly, following Lemma \ref{lemma7} and Lemma \ref{lemma8}, we prove that condition \eqref{condition} implies both the convergence result \eqref{compare} and the functional inequality \eqref{inequality}.

Secondly, the generalized Talagrand inequality \eqref{Talagrand} follows directly from the gradient flow interpolation of inequality in Proposition $1$ of \cite{OV}. For completeness of this paper, we present it here.
Consider the real value function 
\begin{equation*}
\Psi(t)=\mathcal{W}_\gamma(\rho_0,\rho_t)+\sqrt{\frac{2\mathcal{D}_\gamma(\rho_t\|\mu)}{\kappa}},
\end{equation*}
where $\rho_t=\rho(t,\cdot)$ is the density function at time $t$. Notice that $\Psi(0)=\mathcal{W}(\rho_0, \mu)$ and $\lim_{t\rightarrow \infty}\Psi(t)=\sqrt{\frac{2\mathcal{D}_\gamma(\rho_t\|\mu)}{\kappa}}$, since $\mathcal{D}_\gamma(\rho_t\|\mu)\rightarrow 0$ following \eqref{compare}.

We next claim $\frac{d}{dt}\Psi(t)\leq 0$. If so, we finish the proof. To prove it, we show that \begin{equation*}
\frac{d}{dt}|^{+}\Psi(t)=\lim\sup_{h\rightarrow 0}\frac{1}{h}(\Psi(t+h)-\Psi(t))\leq 0.
\end{equation*}
Notice the fact that 
 \begin{equation*}
 |\mathcal{W}_\gamma(\rho_{t+h}, \rho)- \mathcal{W}_\gamma(\rho_{t}, \rho)|\leq \mathcal{W}_\gamma(\rho_{t+h}, \rho_t),
 \end{equation*} 
 and  along the gradient flow $\partial_t\rho=-\textrm{grad}_g\mathcal{D}_\gamma(\rho\|\mu)$, 
 \begin{equation*}
\begin{split}
 \lim\sup_{h\rightarrow 0} \frac{1}{h}\mathcal{W}_\gamma(\rho_{t+h}, \rho_t)=&g_\rho(\partial_t\rho_t, \partial_t\rho_t)\\
=&\sqrt{g_\rho(\textrm{grad}_g\mathcal{D}_\gamma(\rho_t\|\mu), \textrm{grad}_g\mathcal{D}_\gamma(\rho_t\|\mu))}\\
=&\sqrt{\mathcal{I}_\gamma(\rho_t)}.
\end{split}
\end{equation*}
In addition
\begin{equation*}
\begin{split}
\frac{d}{dt}\sqrt{\frac{2\mathcal{D}_\gamma(\rho_t\|\mu)}{\kappa}}=&\sqrt{\frac{2}{\kappa}\frac{1}{\mathcal{D}_\gamma(\rho_t\|\mu)}}\frac{d}{dt}\mathcal{D}_\gamma(\rho_t\|\mu)\\
=&\sqrt{\frac{2}{\kappa}\frac{1}{\mathcal{D}_\gamma(\rho_t\|\mu)}}\Big(-\mathcal{I}_\gamma(\rho_t\|\mu)\Big)\\
=&-\sqrt{\frac{2}{\kappa}\frac{\mathcal{I}_\gamma(\rho_t\|\mu)}{\mathcal{D}_\gamma(\rho_t\|\mu)}}\sqrt{\mathcal{I}_\gamma(\rho_t\|\mu)}\\
\leq& -\sqrt{\mathcal{I}_\gamma(\rho_t\|\mu)}.
\end{split}
\end{equation*}
Thus 
\begin{equation*}
\begin{split}
\frac{d}{dt}|^{+}\Psi(t)=&\lim\sup_{h\rightarrow 0}\frac{\mathcal{W}_\gamma(\rho_{t+h}, \rho_0)-\mathcal{W}_\gamma(\rho_t,\rho_0)}{h}+\frac{d}{dt}\mathcal{D}_\gamma(\rho_t\|\mu)|_{t=0}\\
\leq & \sqrt{\mathcal{I}_\gamma(\rho_t\|\mu)}- \sqrt{\mathcal{I}_\gamma(\rho_t\|\mu)}=0,
\end{split}
\end{equation*}
which finishes the proof.
\end{proof}

\begin{proof}[Proof of Theorem \ref{cor2}]
We prove the equality by using the Hessian operator of $\mathcal{D}_\gamma(\rho\|\mu)$ in $(\mathcal{P}, g)$ at the point $\rho=\mu$. Notice that for any $\sigma\in T_\rho\mathcal{P}$, then 
\begin{equation}\label{pth}
\int \sigma \delta\mathcal{D}_\gamma(\rho\|\mu) dx|_{\rho=\mu}=\int \frac{1}{1-\gamma}(\frac{\rho}{\mu})^{1-\gamma}\sigma dx|_{\rho=\mu}=\frac{1}{1-\gamma}\int \sigma dx=0.
\end{equation}
Following the Hessian operator formula in \eqref{Riem}, denote $\sigma=-\nabla\cdot(\rho^\gamma\nabla\Phi)$, then
\begin{equation*}
\begin{split}
\textrm{Hess}_g\mathcal{D}_\gamma(\rho\|\mu)(\sigma,\sigma)|_{\rho=\mu}=&\int \delta^2\mathcal{D}_\gamma(\rho\|\mu)\sigma^2 dx-\int \delta\mathcal{D}_\gamma(\rho\|\mu)\Gamma_\rho(\sigma,\sigma)dx|_{\rho=\mu}\\
=&\int \delta^2\mathcal{D}_\gamma(\rho\|\mu)\sigma^2 dx|_{\rho=\mu}\\
=&\int \frac{1}{\mu}\Big(\nabla\cdot(\mu^\gamma\nabla\Phi)\Big)^2 dx.
\end{split}
\end{equation*}
where the second equality uses the fact $\Gamma_\rho(\sigma, \sigma)\in T_\rho\mathcal{P}$ and \eqref{pth}. Comparing the above with formula at $\rho=\mu$ in Lemma \ref{lemma7}, we prove the equality. 
\end{proof}

\begin{proof}[Proof of Corollary \ref{GPI}]
We first prove the following claim.

\noindent\textbf{Claim}: 
\begin{equation}\label{claim}
\begin{split}
&\min_{\sigma\in T_\rho\mathcal{P}}\Big\{\textrm{Hess}_g\mathcal{D}_\gamma(\rho\|\mu)(\sigma, \sigma)|_{\rho=\mu}\colon g_\mu(\sigma,\sigma)=1\Big\}\\
=&\min_{\Phi\in C^{\infty}(M)} \Big\{\int \frac{1}{\mu}\Big(\nabla\cdot(\mu^\gamma\nabla\Phi)\Big)^2 dx\colon \int \|\nabla\Phi\|^2\mu^\gamma dx=1 \Big\} \\
=&\min_{f\in C^{\infty}(M)} \Big\{ \int \|\nabla f\|^2 \mu^\gamma dx \colon \int f^2\mu dx=1,\quad \int f\mu dx=0\Big\}.
\end{split}
\end{equation}
\begin{proof}[Proof of Claim]
The first equality holds by the definition of Hessian operator at $\rho=\mu$ as in the proof of Theorem \ref{cor2}. We next focus on the second equality. 
Denote $\sigma_1=-\nabla\cdot(\mu^\gamma\nabla\Phi)$. Then the minimization in the second equation of \eqref{claim} forms 
\begin{equation*}
\lambda_1:=\min_{\sigma_1\in T_\rho\mathcal{P}} \Big\{\int \frac{1}{\mu}\sigma_1^2 dx\colon \int (\sigma_1, -\Delta_{\mu^\gamma}^{-1}\sigma_1) dx=1 \Big\}.
\end{equation*}
The minimizer of above minimization satisfies the following eigenvalue problem
\begin{equation*}
\frac{1}{\mu}\sigma_1 =-\lambda_1\Delta_{\mu^\gamma}^{-1}\sigma_1, 
\end{equation*}
i.e. 
\begin{equation*}
-\nabla\cdot(\mu^\gamma \nabla\frac{\sigma_1}{\mu}) =\lambda_1 \sigma_1.
\end{equation*}
In other words, $\lambda_1=\lambda_\textrm{min}(-\Delta_{\mu^\gamma}\frac{1}{\mu})$, where $\lambda_{\min}$ represents the smallest non-zero eigenvalue.

On the other hand, denote $\sigma_2=f\mu$, then the minimizer of minimization \eqref{claim} in the third equality forms 
\begin{equation*}
\lambda_2:=\min_{\sigma_2\in T_\rho\mathcal{P}}\Big\{ \int \|\nabla\frac{\sigma_2}{\mu}\|^2\mu^\gamma dx\colon \int \frac{\sigma_2^2}{\mu} dx=1\Big\}
\end{equation*}
Similarly, the minimizer of above minimization satisfies the following eigenvalue problem
\begin{equation*}
-\frac{1}{\mu}\nabla\cdot(\mu^\gamma\nabla\frac{\sigma_2}{\mu})=\lambda_2 \frac{\sigma_2}{\mu}
\end{equation*}
i.e. 
\begin{equation*}
-\nabla\cdot(\mu^\gamma\nabla\frac{\sigma_2}{\mu})=\lambda_2\sigma_2. 
\end{equation*}
Thus $\lambda_2=\lambda_\textrm{min}(-\Delta_{\mu^\gamma}\frac{1}{\mu})$. From the above, we have $\lambda_1=\lambda_2$, which finishes the proof of claim.
\end{proof}
From the above claim, the smallest eigenvalue of Hessian operator of $\mathcal{D}_\gamma$ in $(\mathcal{P}, g)$ at $\rho=\mu$ is precisely the lower bound for the Poincar{\'e} inequality. 
Here from the generalized Yano's formula, we have
\begin{equation*}
\begin{split}
&\textrm{Hess}_g\mathcal{D}_\gamma(\rho\|\mu)(\sigma, \sigma)|_{\rho=\mu}\\
=& \int \frac{1}{\mu}(-\Delta_{\mu^\gamma}\Phi)^2 dx\\
=&\int \mu^\gamma \Big\{\big(\mu^{\gamma-1}\textrm{Ric}-\Delta\mu^{\gamma-1}-\frac{1}{\gamma-1}\textrm{Hess}\mu^{\gamma-1}\big)(\nabla\Phi, \nabla\Phi)+\mu^{\gamma-1}\|\textrm{Hess}\Phi\|^2\\
&\hspace{1.2cm}+\gamma(\gamma-1)\mu^{\gamma-1}J(\Phi,\Phi)|_{\rho=\mu}\Big\}dx,
\end{split}
\end{equation*}
where
\begin{equation*}
\begin{split}
J(\Phi,\Phi)|_{\rho=\mu}=&(\nabla\log\mu,\nabla\Phi)^2-\frac{1}{2}(\frac{\nabla\rho}{\rho}, \frac{\nabla\rho}{\rho}+\frac{\nabla\mu}{\mu})(\nabla\Phi, \nabla\Phi)|_{\rho=\mu}\\
=&(\nabla\log\mu, \nabla\Phi)^2-\|\nabla\log\mu\|^2\|\nabla\Phi\|^2.
\end{split}
\end{equation*}
Thus 
\begin{equation*}
-\|\nabla\log\mu\|^2\|\nabla\Phi\|^2\leq J(\Phi,\Phi)|_{\rho=\mu}\leq 0.
\end{equation*}
From the above, we can estimate the smallest eigenvalue of Hessian operator, which finishes the proof.
\end{proof}
\begin{proof}[Proof of Theorem \ref{cor3}]
We first prove the $PH^{-1}I$ inequality. Denote $\rho_t$ be a geodesic curve of least energy 
in $\mathcal{P}$, with $H^{-1}$ metric, where $\rho_0=\mu$ and $\rho_1=\rho$. Then from Proposition \ref{proposition}, $\partial_{tt}\rho_t=0$, i.e. $\rho_t=(1-t)\rho_0+t\rho_1$.
Thus 
\begin{equation*}
H^{-1}(\rho, \mu)=\sqrt{(\rho-\mu, \rho-\mu)_{H^{-1}}}=\sqrt{\int (\rho-\mu, (-\Delta)^{-1}(\rho-\mu))dx}.
\end{equation*}
By taking the Taylor expansion of $\mathcal{D}_0(\rho\|\mu)$ in $(\mathcal{P}, H^{-1})$ at $\rho=\mu$, we obtain  
\begin{equation}\label{f1}
\mathcal{D}_0(\rho\|\mu)=\mathcal{D}_0(\mu\|\mu)+(\textrm{grad}_g\mathcal{D}_0(\rho\|\mu), \rho-\mu)_{H^{-1}}+\int (1-t)\textrm{Hess}_{H^{-1}}\mathcal{F}(\rho_t)(\rho-\mu, \rho-\mu) dt,
\end{equation}
where $\mathcal{D}_0(\mu\|\mu)=0$. From the Cauchy-Schwarz inequality, we have
\begin{equation}\label{f2}
\begin{split}
(\textrm{grad}_g\mathcal{D}_0(\rho\|\mu), \rho-\mu)_{H^{-1}}
\geq&-\sqrt{(\textrm{grad}_g\mathcal{D}_0(\rho\|\mu), \textrm{grad}_g\mathcal{D}_0(\rho\|\mu))_{H^{-1}}}\sqrt{(\rho-\mu, \rho-\mu)_{H^{-1}}}\\
=&-\sqrt{\mathcal{I}_0(\rho\|\mu)} H^{-1}(\rho, \mu).
\end{split}
\end{equation}
In addition, the condition $\mu^{-1}\textrm{Ric}+\textrm{Hess}\mu^{-1}-\Delta \mu^{-1} \succeq \kappa$ implies $\textrm{Hess}_{H^{-1}}\mathcal{D}_0(\rho\|\mu)(\rho-\mu, \rho-\mu)\geq \kappa (\rho-\mu, \rho-\mu)_{H^{-1}}$, thus 
\begin{equation}\label{f3}
\begin{split}
\int (1-t)\textrm{Hess}_{H^{-1}}\mathcal{F}(\rho_t)(\rho-\mu, \rho-\mu) dt\geq &\int_0^1\kappa(1-t) (\rho-\mu, \rho-\mu)_{H^{-1}}dt \\
=&\frac{\kappa}{2}H^{-1}(\rho,\mu)^2.
\end{split}
\end{equation}
Substituting \eqref{f2} and \eqref{f3} into \eqref{f1}, we prove the $PH^{-1}I$ inequality. In addition, the $H^{-1}$-Talagrand inequality follows directly from Theorem \ref{thm}. %Combining the above two inequalities, we prove the last equality, which finishes the proof.
\end{proof}
\begin{remark}\label{rmk9}
The current method fails when $\gamma>1$ or $\gamma<0$. In these cases, there is no finite lower bound for the bilinear form and squared gradient norm for any $\rho\in\mathcal{P}$. One can not obtain the finite ratio between $\frac{d}{dt}\mathcal{D}_\gamma(\rho_t\|\mu)$ and $\frac{d^2}{dt^2}\mathcal{D}_\gamma(\rho_t\|\mu)$. Thus we can not establish the exponential decay results in term of $\gamma$--divergence. 

However, the current method fails does not mean that we can not find the convergence guarantee condition of $\gamma$--diffusion processes when $\gamma>1$. In fact, we can always formulate $\gamma$--divergence as the gradient flow of 1-divergence (relative entropy) w.r.t. density manifold metric $\Big(-\nabla\cdot(\rho\mu^{\gamma-1}\nabla)\Big)^{-1}$. In this case, the study of diffusion hypercontractivity forms a classical Bakry--{\'E}mery method. In other words, one can always apply the entropy method or entropy-entropy production as in \cite{EMethod1} to find the associated diffusion hypercontractivity and convergence rate in 1-divergence. See related details in \cite{Gentil}. 
%In this perspective, our geometric computation in density manifold can help the derivation or calculation in classical entropy method. 
\end{remark}
\begin{remark}
We comment on the proof of different types of inequalities. (i) For Log-Sobolev and Talagrand inequalities, we only need the Hessian operator along the gradient flow to have a lower bound. 
(ii) For Poin{\'c}are inequality, we require the Hessian operator at the equilibrium measure $\mu$ to have a lower bound. (iii) For the divergence, metric and information inequality, such as HWI or P$H^{-1}$I inequality, we require the Hessian operator to have a lower bound for any tangent directions in density manifold. Interestingly, the above three conditions coincide in the case of $\gamma=0,1$.
%E.g. this is true when $\gamma=1,0$ or $\mu=1$.
\end{remark}

\section{Generalized Bakry--{\'E}mery Calculus}\label{section5}
In this section, we propose the generalized Bakry--{\'E}mery iterative calculus. This definition follows the connection of Hessian operator in density manifold with the generator (Kolmogorov backward operator) of $\gamma$--drift diffusion process.

We first define the generalized iterative Bakry--{\'E}mery Gamma operators. 
\begin{definition}[$\gamma$--Bakry--{\'E}mery calculus]
Denote the $\gamma$--Gamma one operator $\Gamma_{\gamma,1}\colon C^{\infty}(M)\times C^{\infty}(M)\times \mathcal{P}\rightarrow C^{\infty}(M)$ by
\begin{equation*}
\Gamma_{\gamma,1}(\Phi_1, \Phi_2,\rho)= (\nabla\Phi_1, \nabla\Phi_2)\rho^{\gamma-1},\end{equation*}
\textrm{where $\Phi_1$, $\Phi_2\in C^{\infty}(M)$.}

Denote the $\gamma$--Gamma two operator $\Gamma_{\gamma,2}\colon C^{\infty}(M)\times C^{\infty}(M)\times \mathcal{P}\rightarrow C^{\infty}(M)$ by
\begin{equation*}
\Gamma_{\gamma,2}(\Phi_1,\Phi_2,\rho)=\frac{\gamma}{2} L_\gamma \Gamma_{\gamma,1}(\Phi_1, \Phi_2,\rho)-\frac{1}{2}\Gamma_{\gamma,1}(\Phi_1, L_\gamma \Phi_2,\rho)-\frac{1}{2}\Gamma_{\gamma,1}(\Phi_2, L_\gamma \Phi_1,\rho),
\end{equation*}
\textrm{where $\Phi_1$, $\Phi_2\in C^{\infty}(M)$.}
\end{definition}
\begin{remark}
We note that when $\gamma=1$, we recover the classical iterative Bakry--{\'E}mery operators. Here $\Gamma_{1,1}$ and $\Gamma_{1,2}$ are independent of $\rho$ with 
\begin{equation*}
\Gamma_{1,1}(\Phi, \Phi)=(\nabla\Phi, \nabla\Phi)\quad\textrm{and}\quad \Gamma_{1,2}(\Phi, \Phi)=\frac{1}{2}L_1\Gamma_{1,1}(\Phi, \Phi)-\Gamma_{1,1}(\Phi, L_1\Phi),
\end{equation*} 
where $L_1=(\nabla \log\mu, \nabla\cdot)+\Delta$ is the generator of classical Langevin drift diffusion process. In addition, when $\gamma\neq 1$, the generalized Bakery--{\'E}mery Gamma one and Gamma two operators depend on the current density $\rho$. In other words, they are mean-field formulations of Gamma operators. 
\end{remark}

We next prove an equality to bridge generalized Bakry--{\'E}mery calculus and Hessian operator of $\gamma$--divergence in density manifold.
\begin{proposition}\label{prop10}
\begin{equation*}
\textrm{Hess}_g\mathcal{D}_\gamma(\rho\|\mu)(\sigma_1,\sigma_2)=\int \Gamma_{\gamma,2}(\Phi_1,\Phi_2,\rho)(x)\rho(x)dx,
\end{equation*}
where $\sigma_i=-\nabla\cdot(\rho^\gamma \nabla\Phi_i)\in T_\rho\mathcal{P},$
and $\Phi_i\in C^\infty(M)$ with $i=1,2$.
\end{proposition}
\begin{proof}
For simplicity of presentation, in the proof, we omit the notation of $\rho$ with the generalized Gamma operators, e.g. $\Gamma_{\gamma,1}(\Phi_1,\Phi_2):=\Gamma_{\gamma,1}(\Phi_1,\Phi_2,\rho)$.

Let us recalculate the Hessian operator of $\mathcal{D}_\gamma$ in $(\mathcal{P}, g)$ by using \eqref{Riem} directly. Using the generalized iterative operators, we reformulate \eqref{Riem} as follows:
\begin{equation}\label{Hessa}
\begin{split}
\textrm{Hess}_g\mathcal{D}_\gamma(\rho\|\mu)(\sigma_1, \sigma_2)=&\int \delta^2\mathcal{D}_\gamma\Big(\nabla\cdot(\rho^\gamma \nabla\Phi_1)\Big)\Big(\nabla\cdot(\rho^\gamma \nabla\Phi_2)\Big)dx\\
&+ \frac{\gamma}{2}\int\Big\{\Gamma_{1,1}(\Gamma_{\gamma,1}(\delta\mathcal{D}_\gamma,\Phi_1), \Phi_2)+\Gamma_{1,1}(\Gamma_{\gamma,1}(\delta\mathcal{D}_\gamma,\Phi_2), \Phi_1)\\
&\hspace{1cm}-\Gamma_{1,1}(\Gamma_{\gamma,1}(\Phi_1,\Phi_2), \delta\mathcal{D}_\gamma)\Big\} \rho^\gamma dx.
\end{split}
\end{equation}
We next rewrite \eqref{Hessa} in three terms. First, we prove the following claim. 

\noindent\textbf{Claim 1:}
\begin{equation}\label{step1}
\begin{split}
&\frac{1}{2}\int\Gamma_{\gamma,1}(\Phi_1, L_\gamma \Phi_2)\rho dx\\=&\frac{1}{2}\int \delta^2\mathcal{D}_\gamma\Big(\nabla\cdot(\rho^\gamma \nabla\Phi_1)\Big)\Big(\nabla\cdot(\rho^\gamma \nabla\Phi_2)\Big)+\gamma\Gamma_{1,1}(\Gamma_{\gamma,1}(\delta\mathcal{D}_\gamma,\Phi_1), \Phi_2)\rho^\gamma dx.
\end{split}
\end{equation}
\begin{proof}[Proof of Claim 1]
Notice
\begin{equation*}
\int \Gamma_{1,1}(\Gamma_{\gamma,1}(\delta\mathcal{D}_\gamma,\Phi_1), \Phi_2)\rho^\gamma dx=-\int \nabla\cdot(\rho^\gamma \nabla\Phi_2) \Gamma_{\gamma,1}(\delta\mathcal{D}_\gamma,\Phi_1) dx.
\end{equation*}
and 
\begin{equation*}
\nabla\cdot(\rho^\gamma \nabla\Phi_1)=(\nabla\rho^\gamma, \nabla\Phi_1)+\rho^\gamma\Delta\Phi_1.
\end{equation*}
The above two facts show that 
\begin{equation*}
\begin{split}
\textrm{R.H.S. of \eqref{step1}}=&\frac{1}{2}\int \nabla\cdot(\rho^\gamma \nabla\Phi_2)\Big\{\delta^2\mathcal{D}_\gamma \nabla\cdot(\rho^\gamma \nabla\Phi_1)-\gamma\Gamma_{\gamma,1}(\delta\mathcal{D}_\gamma,\Phi_1)\Big\} dx\\
=&\frac{1}{2}\int \nabla\cdot(\rho^\gamma \nabla\Phi_2)\Big\{\delta^2\mathcal{D}_\gamma (\nabla \rho^\gamma, \nabla\Phi_1)+ \delta^2\mathcal{D}_\gamma \rho^\gamma\Delta\Phi_1-\gamma(\nabla\delta\mathcal{D}_\gamma,\nabla\Phi_1)\rho^{\gamma-1}\Big\} dx.
\end{split}
\end{equation*}
Using the fact that $\delta^2\mathcal{D}_\gamma=\rho^{-\gamma}\mu^{\gamma-1}$ and $\nabla\delta\mathcal{D}_\gamma= \rho^{-\gamma}\mu^{\gamma-1}\nabla\rho-\rho^{1-\gamma}\mu^{\gamma-2}\nabla\mu$.
\begin{equation*}
\begin{split}
\textrm{R.H.S. of \eqref{step1}}=&\frac{1}{2}\int \nabla\cdot(\rho^\gamma \nabla\Phi_2)\Big\{\rho^{-\gamma}\mu^{\gamma-1}(\nabla \rho^\gamma, \nabla\Phi_1)+ \rho^{-\gamma}\mu^{\gamma-1} \rho^\gamma\Delta\Phi_1\\
&\hspace{2.8cm}-\gamma(\rho^{-\gamma}\mu^{\gamma-1}\nabla\rho-\rho^{1-\gamma}\mu^{\gamma-2}\nabla\mu,\nabla\Phi_1)\rho^{\gamma-1}\Big\} dx\\
=&\frac{1}{2}\int \nabla\cdot(\rho^\gamma \nabla\Phi_2)\Big\{\mu^{\gamma-1}\Delta\Phi_1+\gamma(\mu^{\gamma-2}\nabla\mu,\nabla\Phi_1)\Big\} dx\\
=&\frac{1}{2}\int \nabla\cdot(\rho^\gamma \nabla\Phi_2) L_\gamma\Phi_1 dx\\
=&-\frac{1}{2}\int (\nabla L_\gamma\Phi_1, \nabla\Phi_2)\rho^\gamma dx\\
=&-\frac{1}{2}\int \Gamma_{\gamma,1}(L_\gamma\Phi_1, \Phi_2)\rho dx,
\end{split}
\end{equation*}
where the second last equality holds by the integration by parts formula. 
\end{proof}

Secondly, by switching $\Phi_1$ and $\Phi_2$ in Claim 1, we have 
\begin{equation}\label{step2}
\begin{split}
&\frac{1}{2}\int\Gamma_{\gamma,1}(\Phi_2, L_\gamma \Phi_1) \rho dx\\=&\int \frac{1}{2}\delta^2\mathcal{D}_\gamma\Big(\nabla\cdot(\rho^\gamma \nabla\Phi_1)\Big)\Big(\nabla\cdot(\rho^\gamma \nabla\Phi_2)\Big)+\frac{\gamma}{2}\Gamma_{1,1}(\Gamma_{\gamma,1}(\delta\mathcal{D}_\gamma,\Phi_2), \Phi_1)\rho^\gamma dx.
\end{split}
\end{equation}

Thirdly, we show the following claim. 

\noindent\textbf{Claim 2:}
\begin{equation}\label{step3}
\begin{split}
\int\Gamma_{\gamma,1}(\Phi_1, L_\gamma \Phi_2) \rho dx=\int \Gamma_{1,1}(\Gamma_{\gamma,1}(\Phi_1,\Phi_2), \delta\mathcal{D}_\gamma) \rho^\gamma dx.
\end{split}
\end{equation}
\begin{proof}[Proof of Claim 2]
Here
\begin{equation*}
\begin{split}
\textrm{R.H.S. of \eqref{step3}}=&\int \Gamma_{1,1}(\Gamma_{\gamma,1}(\Phi_1,\Phi_2), \delta\mathcal{D}_\gamma) \rho^\gamma dx\\
=& -\int \nabla\cdot(\rho^\gamma \nabla\delta\mathcal{D}_\gamma) \Gamma_{\gamma,1}(\Phi_1,\Phi_2) dx\\
=& -\int  L^*_\gamma\rho \Gamma_{\gamma,1}(\Phi_1,\Phi_2) dx\\
=&\int L_\gamma\Gamma_{\gamma,1}(\Phi_1,\Phi_2)\rho dx,
\end{split}
\end{equation*}
where the second equality uses the fact that 
\begin{equation*}
\nabla\cdot(\rho^\gamma \nabla\delta\mathcal{D}_\gamma)=\nabla\cdot(\mu^\gamma \nabla\frac{\rho}{\mu})=L^*\rho,
\end{equation*}
and the last equality holds because $L^*_\gamma$ is the adjoint operator $L_\gamma$ in $L^2(\rho)$.
\end{proof}

By summing \eqref{step1}, \eqref{step2} and $\frac{\gamma}{2}$ times \eqref{step3} and using \eqref{Hessa}, we have
\begin{equation*}
\begin{split}
&\textrm{Hess}_g\mathcal{D}_\gamma(\rho\|\mu)(\sigma_1, \sigma_2)\\
=&\int \Big\{-\frac{1}{2}\Gamma_{\gamma,1}(\Phi_1, L_\gamma \Phi_2)-\frac{1}{2}\Gamma_{\gamma,1}(\Phi_2, L_\gamma \Phi_1)+\frac{\gamma}{2} L_\gamma \Gamma_{\gamma,1}(\Phi_1, \Phi_2) \Big\}\rho(x)dx\\
=& \int \Gamma_{\gamma,2}(\Phi_1,\Phi_2)(x)\rho(x)dx.
\end{split}
\end{equation*}
\end{proof}

We last point out that the generalized Bakry--{\'E}mery iterative calculus implies generalized hypercontractivity. 
\begin{proposition}[Generalized Bakry--{\'E}mery criterion]
If there exists a constant $\kappa>0$, such that
\begin{equation}\label{a}
 \int \Gamma_{\gamma, 2}(\Phi,\Phi,\rho)(x)\rho(x)dx \geq \kappa \int \Gamma_{\gamma,1}(\Phi,\Phi,\rho)(x)\rho(x)dx, 
\end{equation}
 for $\Phi=\frac{1}{1-\gamma}(\frac{\rho}{\mu})^{1-\gamma}$, with respect to any $\rho\in \mathcal{P}$.
Then the generalized hypercontractivity \eqref{hyper} and the generalized Log-Sobolev inequality \eqref{inequality} hold. 
\end{proposition}
\begin{remark}
Our generalized Bakry-{\'E}mery operators follow the proof in proposition 19 of \cite{LiG1}. In other words, when the divergence functional is the relative entropy, i.e. $\gamma=1$, we have the classical Bakry-{\'E}mery iterative calculus. 
For generalized divergence functional, we introduce the generalized iterative Bakry-{\'E}mery calculus.
\end{remark}
\begin{remark}
When $\gamma=1$ or $\gamma=0$, the ratio between generalized Gamma two operator and Gamma one operator gives the bound in \eqref{a}. This is not the case for $\gamma\neq 1,0$. In general, we need to apply the mean field (integral formula w.r.t $\rho$) of the Gamma two operator to bound the Gamma one operator, and then derive the related Log-Sobolev inequalities. 
\end{remark}
\begin{proof}
Here the proof applies Proposition \ref{prop10} and the gradient flow formulation \eqref{compare} to prove Theorem \ref{thm}.  
\end{proof}
As a summary, we show the generalized Bakry--{\'E}mery criterion \eqref{a} here, and estimate its precise bound in Theorem \ref{thm}. Besides, we comment on major differences and difficulties between generalized Bakry--{\'E}mery criterions and classical ones. The Hessian operator in generalized density manifold involves an additional quadratic form $J(\Phi, \Phi)$.  Thus the smallest eigenvalue of Hessian operator in density manifold is not enough to provide a lower bound for the convergence rate of generalized drift-diffusion processes. Here, we carefully derive the global behavior of dynamics. This is to control the additional quadratic form along with the gradient flow for any $\rho$. Besides, a local viewpoint is provided for establishing the Poincar{\'e} inequality, which follows local behavior of dynamics, i.e., the Hessian operator in density manifold at the minimizer $\mu$. This local property relates to the integral formula, known as the Yano's formula.

\section{Discussion}
In this paper, we study the diffusion hypercontractivity for $\gamma$--drift-diffusion process, and prove generalized Log-Sobolev, Poincar{\'e} and Talagrand inequalities. 
Firstly, using $\mathcal{D}_\gamma$ as the Lyapunov function, the global exponential convergence of $\gamma$--drift-diffusion process is presented for $\gamma\in [0,1]$. It is to estimate the smallest eigenvalue of Hessian operator in density manifold along with the gradient flow for any $\rho\in \mathcal{P}$. Secondly, the local behavior of $\gamma$-drift diffusion process is shown for any $\gamma\in \mathbb{R}$. It is to estimate the smallest eigenvalue of Hessian operator at the reference measure $\mu$. This local property allows us to obtain a class of generalized Poincar{\'e} inequalities. Besides, we identify the generalized Poincar{\'e} inequality and Yano's formula. Lastly, our approach can be formulated into the generalized Bakry--{\'E}mery iterative operators. Here, the Gamma one and Gamma two operators are mean-field based, which depend on the current density; see related studies in probability models \cite{LiG}. In future work, we will study more general mean-field Bakry--{\'E}mery conditions for related diffusion hypercontractivity and functional inequalities. 

\newpage\section*{Notations}
We apply the following notations in this paper.

\begin{tabular}{|l|c|}
\hline 
Base manifold & $M$\\
Metric & $(\cdot,\cdot)$\\
Norm & $\|\cdot\|$\\
Divergence operator & $\nabla\cdot$\\
Gradient operator & $\nabla$\\
Hessian operator & $\textrm{Hess}$\\
\hline
Density manifold & $\mathcal{P}$ \\
Probability density & $\rho$ \\
Reference density & $\mu$\\
Tangent space & $T_\rho\mathcal{P}$ \\
Cotangent space & $T_\rho^*\mathcal{P}$\\
Density manifold metric tensor & $g_\rho$\\
Weighted Laplacian operator & $\Delta_h=\nabla\cdot(h\nabla)$\\
First $L^2$ variation & $\delta$\\
Second $L^2$ variation &$\delta^2$ \\ %$d^2_{pp}$
Gradient operator & $\textrm{grad}_g$\\
\textrm{Hess}ian operator &$\textrm{\textrm{Hess}}_g$ \\
Christoffel symbol & $\Gamma_\rho(\cdot, \cdot)$\\
Tangent bundle & $(\rho, \sigma)\in T\mathcal{P}$\\
Cotangent bundle & $(\rho, \Phi)\in T^*\mathcal{P}$\\
\hline
$\gamma$--Divergence & $\mathcal{D}_\gamma$\\
$\gamma$--Fisher information &$\mathcal{I}_\gamma$\\
$\gamma$--Wasserstein distance & $\mathcal{W}_\gamma$\\
$\gamma$--Diffusion process generator & $L_\gamma$\\
$\gamma$--Gamma one operator & $\Gamma_{\gamma, 1}$\\
$\gamma$--Gamma two operator & $\Gamma_{\gamma, 2}$\\
\hline
Log Sobolev constant & $\kappa$\\
Poincar{\'e} constant & $\lambda$\\
\hline 
\end{tabular}
\end{document}